\documentclass[a4paper,UKenglish,cleveref, autoref, thm-restate]{lipics-v2021}

\usepackage{microtype}
\usepackage{latexsym,epsfig,graphicx,amsmath,algorithm,algorithmicx,algpseudocode,graphicx,epsfig,multirow,calc,colortbl}
\usepackage{amsthm}
\usepackage{amsmath}
\usepackage{mathtools}
\usepackage{empheq}
\usepackage{cases}
\usepackage{color}
\usepackage{cite}

\theoremstyle{plain}
\def\max{{\rm max}}
\def\min{{\rm min}}
\long\def\longdelete#1{}

\newcommand{\oural}{$\mathcal{A}${}}

\newtheorem{clm}[theorem]{\textbf{Claim}}
\newcommand{\ps}{\mathrm{\psi_\sigma}}

\title{Improving the Bounds of the Online Dynamic Power Management Problem} 

\author{Ya-Chun Liang}{Department of Industrial Engineering and Engineering Management,
National Tsing Hua University, Hsinchu 30013, Taiwan}{ycliang512@gapp.nthu.edu.tw}{}{}

\author{Kazuo Iwama\footnote{Supported by MOST Taiwan under Grants
110-2223-E-007-001, 110-2811-E-007-507 and KAKENHI Japan under Grant 16H02782.}}{Department of Industrial Engineering and Engineering Management,
National Tsing Hua University, Hsinchu 30013, Taiwan}{iwama@ie.nthu.edu.tw}{}{}

\author{Chung-Shou Liao\footnote{Supported by MOST Taiwan under Grants 110-2221-E-007-106-MY3 and 110-2622-8-007-017.}}{Department of Industrial Engineering and Engineering Management, 
National Tsing Hua University, Hsinchu 30013, Taiwan}{csliao@ie.nthu.edu.tw}{}{}

\authorrunning{Y.-C. Liang, K. Iwama, C.-S. Liao} 

\Copyright{Ya-Chun Liang, Kazuo Iwama, Chung-Shou Liao} 



\category{} 

\relatedversion{} 

\nolinenumbers 

\longdelete{
\EventEditors{John Q. Open and Joan R. Access}
\EventNoEds{2}
\EventLongTitle{42nd Conference on Very Important Topics (CVIT 2016)}
\EventShortTitle{CVIT 2016}
\EventAcronym{CVIT}
\EventYear{2016}
\EventDate{December 24--27, 2016}
\EventLocation{Little Whinging, United Kingdom}
\EventLogo{}
\SeriesVolume{42}
\ArticleNo{23}
}

\begin{document}

\maketitle

\renewcommand{\algorithmicrequire}{\textbf{Input:}}
\renewcommand{\algorithmicensure}{\textbf{Output:}}

\newcommand{\il}{I\kern-.07em\ell}
\newcommand{\ir}{Ir}
\newcommand{\gf}{G\kern-.1emf}
\newcommand{\gl}{G\ell}
\newcommand{\gr}{Gr}
\newcommand{\ol}{O\ell}
\newcommand{\of}{Of}
\newcommand{\orr}{Or}
\newcommand{\hll}{H\ell}
\newcommand{\hrr}{Hr}

\newcommand{\cc}{\color{red}\ }
\newcommand{\ccc}{\color{black}\ }

\begin{abstract}
We investigate the {\em power-down mechanism} which decides when a machine transitions  
between states such that the total energy consumption,
characterized by 
execution cost, idle cost and switching cost,
is minimized. In contrast to most of the previous studies on the offline model, we focus on the online model
in which a sequence of jobs with their release time, execution time and deadline, arrive in an online fashion.
More precisely, 
we exploit a different switching on and off strategy and 
present an upper bound of 3, and further show a lower bound of 2.1, in a dual-machine model,
introduced by Chen et al. in 2014 [STACS 2014: 226-238], both of which beat the currently best result.     
\end{abstract}

\section{Introduction}
\label{intro}

\longdelete{ 
Machines consume energy. 
So the easiest way for a company to reduce the cost is to reduce the number of machines, such as cars, computers, copy machines, air conditioners, TV's, fridges, etc., installed in the company. 
If one office has two copy machines, it is often directed to give up one of them. 
Thus it appears to be indisputable that the number of machines should be decreased to save energy and associated cost. 
A bit surprisingly, however, this is not absolutely true; {\em increasing} the number of machines does reduce the cost in a certain~situation.}

Machines consume energy and many of them have the following property: namely machines have two states, ON and OFF, and change between the two states quite often.
Furthermore, it requires a relatively large amount of energy to change from OFF to ON. 
For instance, a laptop goes to sleep or shutdown (OFF) if we do not touch it for, 
e.g.,
one hour, and comes back to ON if we press some key,
but it is actually energy consuming.
Copy machines also turn off automatically if there is no job for, e.g., ten minutes. 
Similarly, it consumes relatively large energy for heating them up when a new job arrives. 
There are many other similar examples, and it is obviously important to design online algorithms to minimize the energy consumption for this type of machines. 
Fortunately most of those problems are essentially equivalent to the well-known \emph{ski-rental} problem and the optimal solution has been known for a long time both in deterministic and randomized cases~\cite{10.5555/241938.241951,karlin1988competitive}.

In this study, we investigate the online model of the dynamic power management problem, discussed by 
Irani et al.~\cite{irani2003online,irani2007algorithms} and Chen et al.~\cite{chen2015online},
which aims at determining when to transition a machine between the states:
\emph{busy} or \emph{idle} (ON) and \emph{sleep} (OFF) to finish all input jobs 
such that the energy consumption is minimized. 
Suppose a machine $M$ 
requires $E$ units of energy to change its state from OFF to ON and needs $\psi_\sigma$ units per unit of time to keep it ON, 
i.e. the \emph{idle} cost. 
The goal is to finish all input jobs, arriving in an online fashion, while minimizing the total energy consumption. 
Let us consider a simpler model first, where jobs arriving in an online fashion are 
required to be immediately performed without any delay.
One can observe that the optimal deterministic online algorithm of the ski rental problem can be applied to solve this model, with the competitive ratio of 2.
The intuition is quite straightforward, by turning off the machine $M$ $E/\psi_\sigma$ units of time after its last job is finished.
Note that the performance of an online algorithm is typically evaluated by competitive analysis. 
More precisely,
the quality of an online algorithm is measured by the worst case ratio, called \emph{competitive ratio}, 
which is defined to be the fraction of the cost of the online algorithm over that of the offline optimal algorithm, where the offline algorithm is aware of all jobs in advance.

\longdelete{
The optimal deterministic online algorithm $A$ is to turn off $M$ $E/\psi_\sigma$ units of time after its last job is finished and its competitive ratio (CR) is at most 2. Suppose the value $E/\psi_\sigma$ is 10 minutes for simplicity. 
The worst case is that jobs come every $10+\epsilon$ minutes; 
$A$ needs 20 units of energy (10 units for $E$ and 10 units for its
idle cost)\ccc per job but the optimal algorithm (OPT) $10(1+\epsilon)$\ccc units by keeping $M$ on all the time. 
More sophisticated models are also known, for instance, there is
an intermediate state between {\em idle} and {\em busy}, called {\em stand-by}, which requires a bit more energy to be
kept but relatively less energy cost to switch to ON. 
However, the model does not seem to change dramatically.

A significant change does happen in the above model if the execution of each job is allowed to be postponed.
}

However, when considering a more complicated model, in which each job $j$ is given as $(a_j,d_j,c_j)$, 
where $a_j$ denotes the arrival time, 
$d_j$ the deadline and $c_j$ the execution time, 
a dramatic change does happen.
In other words, the execution of each job is allowed to be postponed, 
which obviously increases the problem hardness.
There have been actually not many studies in the literature for the online model~\cite{10.5555/241938.241951}. 
Irani et al.~\cite{irani2003online,irani2007algorithms} initiated the study of combining the above power-down mechanism with \emph{dynamic speed scaling}, 
where the latter technique has been widely explored in the past decades~\cite{yao1995scheduling,bansal2007speed,chen2007procrastination,chen2006procrastination,irani2007algorithms}.
The basic concept of dynamic speed scaling is that a machine's processing speed can be adjusted dynamically, and the power consumption rate is usually represented by a convex function of the processing speed in terms of time expense.
They proposed the first online algorithm with a constant competitive ratio of $\max \{c_1c_2+c_1+2,4\}$, 
where $c_1$ and $c_2$ are some constant parameters in a given convex power function~\cite{irani2007algorithms}. 
For example, if the power function is quadratic, the upper bound of competitive ratio is~8.
Chen et al.~\cite{chen2015online} considered 
a similar model but without using dynamic speed scaling, 
where an online algorithm can use two machines $M_1$ and $M_2$
instead, under the same assumption that 
all input jobs must be 
finished by the 
offline optimal scheduler using a single machine.
This assumption is actually 
the so-called \emph{single machine schedulability}
condition~\cite{chetto1989scheduling}, which is characterized by the following lemma.
\begin{lemma}\label{assmpution_chen}
\emph{(Chetto et al.~\cite{chetto1989scheduling})}
For any set of jobs $\mathcal{J}$,
they can be optimally scheduled on one machine using the \emph{earliest-deadline-first}(EDF) principle.
That is,
the job with the earliest deadline is always selected for execution at any moment,
if and only if the following condition holds:
\begin{align}\label{input}
    \mbox{For any time interval $(\ell,r)$, we have } \sum_{j:j\in \mathcal{J}, \ell \leq a_j, d_j\leq r}c_j \leq r-\ell.
\end{align}
\end{lemma}
The condition is equivalent to the {\em earliest-deadline-first (EDF)} schedulability~\cite{yao1995scheduling}. 
Here, we also remark that a feasible solution can be a preemptive schedule 
but the jobs can only be executed without migration. 
Chen et al.~\cite{chen2015online}
gave an upper bound of 4 and a lower bound of 2.06 for the competitive ratio of this dual-machine problem.
It is a significant contribution for online dynamic power management problems, but 
unfortunately, the above gap is not very small, for which there has been no improvement up to now.

\longdelete{
Their contribution to the online community is not small at all
for dynamic power management problems.
As mentioned before, the model covers a wide range of real-world problems and the fact that more machines consume less energy sounds quite nontrivial.
Unfortunately, the above gap is not very small, 
for which there has been no improvement up to now.}

\longdelete{
Now each job $j$ is given as $(a_j,d_j,c_j)$, 
where $a_j$ denotes the arrival time, 
$d_j$ the deadline and $c_j$ the execution time.
Fix an arbitrary online algorithm $A$.
Then the adversary lets the turn-on cost $E=10$ and the idle cost $\psi_\sigma=1$, 
and gives $j_1=(0,100,1)$. 
If $A$ does not execute it for $\epsilon$ minutes,
the adversary would release $j_0=(0,100,99-\epsilon/2)$ and $A$ cannot
finish $j_0$ since $100-\epsilon$ minutes left after the delay are not enough
to execute the two jobs whose total amount of execution is $1+99-\epsilon/2$.
However, OPT can by executing $j_1$ exactly when it arrives. ($A$ is
deterministic, so the adversary can figure out $A$'s actions in advance
without looking at $A$'s real actions.)
Thus $A$ must be designed to execute $j_1$ immediately. 
Then the adversary gives $j_2=(10,100,1)$, $j_3=(20,100,1)$, $\ldots$, and $j_{10}=(90,100,1)$.
$A$ must execute all of them immediately once they arrive,
but OPT performs each of them as late as possible, i.e. almost at the end of their deadline. 
One can easily see that $A$ consumes roughly 100 if it keeps ON throughout and OPT 10. 
(It is not hard to observe that $A$ would consume even more energy if
it turns off and turns on the machine somewhere.)
The CR is now 10 and it can be worse unlimitedly by modifying the parameters. 
Thus there is no online
algorithm with a bounded CR.

In this study, we thus investigate the online model, introduced by Chen, Kao, Lee, Rutter and Wagner~\cite{chen2015online},}

\longdelete{
The previous adversary is no longer that scary for an online algorithm $A$. 
For instance, $A$ can simply delay $j_1=(0,100,1)$ as much as possible, 
and when $j_0=(0,100,99-\epsilon/2)$ arrives,
it turns on the second machine, $M_2$, to execute it immediately, 
and $j_1$ is executed on $M_1$, 
resulting in almost the same energy consumption as OPT, 
excluding the turn-on cost of $M_2$.
One can observe that $A$ also spends almost the same energy as OPT against $j_1, j_2, \ldots$. 
Thus an online algorithm would achieve almost the same performance as OPT at least against such a simple input instance.
However, this perfect case does not appear to hold generally.}

\medskip
{\bf Our Contribution.}  
This paper improves both upper and lower bounds to 3 and 2.1, respectively, for exactly the same dual-machine model.
Namely the competitive ratio gap between lower and upper bounds is improved from 1.94 in~\cite{chen2015online} to 0.9.
Our algorithm has the same basic structure as the one in~\cite{chen2015online}, 
but two critical differences, 
which exactly contributes the improvement, 
should deserve being mentioned. 
First, we delay jobs but turn on a machine earlier than its due time by a margin instead of ``\emph{energy-efficient anchor}'' introduced in~\cite{chen2015online}.
Second, 
our idle time after the machine is finished with its execution is not set to $\mathcal{B}=E / \psi_\sigma$, i.e. 
the so-called \emph{break-even time},
but set to twice that value. 
Note that the break-even time has always been used for the class of similar problems including the famous ski-rental problem, 
as mentioned earlier. 
We hope this escape from the common tradition will help on several different occasions in the future.

Moreover,
we use a standard math induction for the analysis, 
making our analysis significantly simpler than that of~\cite{irani2007algorithms, chen2015online}, 
which will also contribute to further improvement of the bounds hopefully.
In~\cite{irani2007algorithms, chen2015online}, 
they made their competitive analysis by introducing what they call a ``sleep interval''
where the optimal offline schedule is OFF and an ``awaken-interval'' where their online algorithm is ON, 
and gave a key lemma saying a single sleep interval overlaps with at most two awaken intervals. This is clever since the analysis boils down to comparing a single (consecutively ON) interval of the online algorithm and that of the optimal offline schedule. 
However, a single such interval of the optimal offline schedule can still contribute to two such intervals of the online algorithm even
if it can actually contribute to only one in many cases. 
Thus the analysis underestimates the cost of the optimal offline schedule.
Our analysis splits the time line simply into ``\emph{phases}'' that realizes the same intervals for an online algorithm and the optimal offline schedule,
which makes it possible to use the standard math induction. 
Of course an online algorithm and the optimal offline schedule can execute different jobs in each phase, 
but that can be managed by considering ``assets'' and ``liabilities'' of jobs between phases.
More details of the basic idea will be given in Sections~\ref{UB} and~\ref{analysis}.

\longdelete{
We remark that in~\cite{irani2007algorithms, chen2015online},
they showed the competitive analysis by considering the \emph{time intervals} in which 
the status of the optimal offline schedule is OFF but 
the status of their online algorithms is ON. 
The analysis may require a slightly larger bound on the actual consumption
due to the definition of time intervals (which will be introduced in the appendix).
Inspired by the design of their time
intervals~\cite{irani2007algorithms, chen2015online}, we consider the time line in a different manner. More precisely, we change to split the time line of our online algorithm separately into \emph{phases} and compare the consumption difference against the corresponding OPT within each single phase.\ccc
}

\medskip
{\bf Other Studies.} 
In comparison with a few amount of research on \emph{pure} dynamic power management, there have been relatively more studies which incorporate speed scaling into the power-down mechanism. 
Albers et al.~\cite{albers2014race} investigated the offline setting of speed scaling with a sleep state and presented a $\frac{4}{3}$-factor approximation algorithm.
Antoniadis et al.~\cite{antoniadis2015fully} further improved the result to a fully-polynomial time approximation scheme.
Considering the multi-machine systems, Demaine et al.~\cite{demaine2013scheduling} developed a polynomial-time algorithm based on dynamic programming.
Albers et al.~\cite{albers2015multi} then considered dynamic speed scaling with job migration.
Very recently, Antoniadis et al.~\cite{antoniadis2020parallel} presented a pseudo-polynomial time algorithm for a single machine based on a linear programming relaxation and a constant-factor approximation algorithm for the case of multiple machines.
Readers may refer to 
the survey works~\cite{benini2000survey,irani2005algorithmic,wierman2012speed} for more details. 

\longdelete{
One may consider the aforementioned model to be a variation of the dynamic power management problem 
which aims at determining when to transition a machine between the states:
active and idle (ON) and sleep (OFF) to finish all input jobs 
such that the energy consumption is minimized. 
There have been actually not many studies in the literature for its online model which can be transformed into the continuous version of the ski-rental problem~\cite{10.5555/241938.241951}, 
if there is only a single machine with two states. 
Irani et al.~\cite{irani2003online,irani2007algorithms} initiated the study of combining the above power-down mechanism with \emph{dynamic speed scaling}, 
where the latter technique has been widely explored in the past decades~\cite{yao1995scheduling,bansal2007speed,chen2007procrastination,chen2006procrastination,irani2007algorithms}.
The basic concept of dynamic speed scaling is that a machine's processing speed can be adjusted dynamically, and the power consumption rate is usually represented by a convex function of the processing speed in terms of time expense.
They proposed the first online algorithm with a constant CR of $\max \{c_1c_2+c_1+2,4\}$, 
where $c_1$ and $c_2$ are some constant parameters in a given convex power function~\cite{irani2007algorithms}. 
For example, if the power function is quadratic, the upper bound of CR is~8.
}

\section{Upper Bound}
\label{UB}

We formally introduce 
the dual-machine model proposed by Chen et al.~\cite{chen2015online}.
Each machine requires a constant amount of energy to switch its state from {\em sleep} to either {\em busy} or {\em idle}, denoted by $E$.
Suppose 
every machine uses a constant speed to execute jobs; that is, 
it consumes a constant amount of energy per unit of time,
denoted by $\psi_b$ when it is busy and $\psi_\sigma$ when it is idle.
The break-even time, 
as mentioned earlier, 
denoted by $\mathcal{B}$,
is defined to be $E / \psi_\sigma$, 
which represents that
when a machine is idle for $\mathcal{B}$ units of time,
the energy consumption is equal to $E$,
i.e. the energy consumption for switching on a machine from sleep to the other states.
The standard assumption is that $\psi_\sigma \leq \psi_b$.
Another assumption in this study is that $E=\psi_b=1$ (and $\ps\leq 1$). 
If $E$ is a positive integer $k$, we can use our model by changing a
unit time from 1 to $1/k$ (and changing the job and idle length
accordingly). Thus our model does not lose any generality and this
setting contributes to significantly simplified expositions.\ccc
Recall that our input always satisfies Condition~(\ref{input}) in this model, 
where a feasible schedule allows job preemption but no migration.

Before introducing our online algorithm, 
we have some more notation.
The dual-machine model has two machines: 
one of them is called $M_P$ (a {\em primary machine}) and the other $M_S$ (a {\em secondary machine}).
We basically use the primary machine as far as possible and turn on the secondary machine if necessary.
That is, we set $M_P$ to be the main machine, 
and $M_S$ to be the backup machine.
When a job arrives,
we always turn on $M_P$ first and decide to turn on $M_S$ only if it finds out that $M_P$ is not able to finish all the pending jobs before their deadlines;
i.e., $M_P$ is overloaded.
In other words, $M_P$ and $M_S$ are logical names and the two (physical) machines may swap their names in our algorithm.
Let $Q_{M_P}$ and $Q_{M_S}$ denote the job queues for $M_P$ and $M_S$, respectively, 
which are already scheduled using EDF; 
namely, the jobs in $Q_{M_P}$ and $Q_{M_S}$ are 
being executed continuously on $M_P$ and $M_S$, respectively. 
We use $Q$ to denote a queue for jobs not yet scheduled. 
If a job in $Q_{M_P}$ or $Q_{M_S}$ is finished, 
it is removed from the queue.
Also, let $c'_j(t)$ denote the remaining execution time of job $j$ at time $t$,
and $W(t,t^\dagger)=\sum_{j\in Q(t,t^\dagger)}c'_j(t)$ denote the total remaining execution time in time interval $(t,t^\dagger)$,
where $Q(t,t^\dagger)$ denotes the subset of jobs that have not yet finished their execution up to time $t$ while their deadlines are not larger than $t^\dagger$, where~$t^\dagger > t$.
When a new job $j'$, given by $(a_{j'},d_{j'},c_{j'})$, is arriving, 
there are two cases. 
One is that the current schedule for $Q_{M_P}$ can accommodate its execution. 
If so, $j'$ is inserted into $Q_{M_P}$,
which is rescheduled by EDF (we say $M_P$ is available).
The other case is that there is no room for $j'$ in $Q_{M_P}$
(if jobs in $Q_{M_P}$ would not have been delayed, 
this cannot happen because of the single machine schedulability
assumption).
We call such a $j'$ {\em urgent}.
If an urgent job comes, then $M_S$ is turned on and $j'$ is executed on $M_S$ immediately 
(where $M_P$ continues executing the already assigned jobs). 
We then swap $M_P$ and $M_S$ at this point of time, which means the original $M_S$ can act as a new primary machine. 
Table~\ref{algo} shows our algorithm, denoted by \oural. 
Note that in the following, 
we use ALG to denote an online algorithm, OPT an offline optimal scheduler and CR a competitive ratio for simplicity.

\begin{table}[h]
\caption{Algorithm \oural}
\label{algo}
\vspace{-18pt}
\begin{center}
\begin{tabular}{l}
\hline
\noalign{\smallskip}
    Initially assign\ccc $M_P$ and $M_S$ to two machines arbitrarily.\\
    The input satisfies Condition~\ref{input} of Lemma~\ref{assmpution_chen}.\\
    At any time $t$, \oural\ proceeds as follows:\\
\noalign{\smallskip}
\hline
\noalign{\smallskip}
\begin{minipage}{5.5in}
    \vskip 2pt
    \begin{enumerate}
      \item Execution of jobs:
      \begin{enumerate}
        \item No machines are on.\\ 
        If there exist $t^{\dagger}$ and $t^{\ast}$ such that $W(t^{\dagger},t^{\ast})\geq t^{\ast}-t^{\dagger}$ and $t\geq t^{\dagger}-u$, turn on one machine ($M_P$) and move all jobs in $Q$ to $Q_{M_P}$.
        \item Only $M_P$ is on and a new job $j$ comes.\\
        If there is no $t^{\ast}$ such that $W(t,t^{\ast})>
        t^{\ast}-t$ for jobs in $Q_{M_P} \cup \{j\}$ (i.e., $M_P$ is available), add $j$ to $Q_{M_P}$ and reschedule it.\\
        Otherwise, turn on the other machine ($M_S$) and add $j$
        to $Q_{M_S}$ that is empty.\\
        After that switch $M_P$ and $M_S$.
	    \item $M_S$ is on and a new job $j$ comes.\\
    	If $M_S$ is available, add $j$ into $Q_{M_S}$ and reschedule it.\\
        Otherwise move it to $Q_{M_P}$ (it is guaranteed that $M_P$ is available).
      \end{enumerate}
      \item Idle state:
      \begin{enumerate}
        \item $M_S$ never has an idle state.\\
        That is, we immediately turn off $M_S$ once $Q_{M_S}$ becomes empty.
        \item $M_P$ becomes idle once $Q_{M_P}$ becomes empty.\\ 
        We turn off $M_P$ when its total idle time becomes $2/\ps$.\\
      \end{enumerate}
    \end{enumerate}
    \vskip 2pt
\end{minipage}
\vspace{-5pt}
\\
\hline
\end{tabular}
\end{center}
\end{table}

\longdelete{
\subsection{Algorithm $\textbf{S}$~\cite{chen2015online} Revisited}

Chen et al.~\cite{chen2015online} presented an online algorithm, 
called \emph{Algorithm $\textbf{S}$}.
Based on the previous assumption that a single machine can finish all the input jobs in the offline model,
Algorithm $\textbf{S}$ basically uses one machine as far as possible and uses the second process if necessary.
That is, they set $M_1$ to be the main machine, 
and $M_2$ to be the backup machine.
When a job arrives,
Algorithm $\textbf{S}$ always turn on the main machine $M_1$ first and decides to turn on the backup machine $M_2$ only if it finds out that $M_1$ is not able to finish all the pending jobs before their deadlines;
i.e., $M_1$ is overloaded.
More precisely, in order to determine when the main machine $M_1$ should be switched on for a pending job $j$,
let $j$ be associated with a parameter $h_j = \max\{a_j, d_j- \lambda\mathcal{B}\}$,
where $\lambda$ is a constant (setting to be one in~\cite{chen2015online}).
The parameter $h_j$ is referred to be the \emph{energy-efficient anchor} for the job $j$.

The key steps of Algorithm $\textbf{S}$ are summarized as follows.
When a new job $j$ arrives and it turns out that $M_1$ cannot finish all the pending jobs before their deadlines if $j$ is assigned to $M_1$,
then switch the state of $M_2$ to ON.
Let the current time be $t^*$ and schedule the jobs that arrived before $t^*$ on $M_1$ based on the earliest-deadline-first (EDF) principle,
and assign the jobs that arrive after $t^*$ to $M_2$.
If both processors $M_1$ and $M_2$ are ON, then Algorithm $\textbf{S}$ decides to turn off $M_1$ at the moment when it becomes idle.
If only one of $M_1$ and $M_2$ is ON, say $M_2$, then when $M_2$ becomes idle, it is switched off once the length between the current time and the time when $M_1$ is turned on has reached $\mathcal{B}$.

\begin{figure}[htb]
\centering
\includegraphics[scale=0.65]{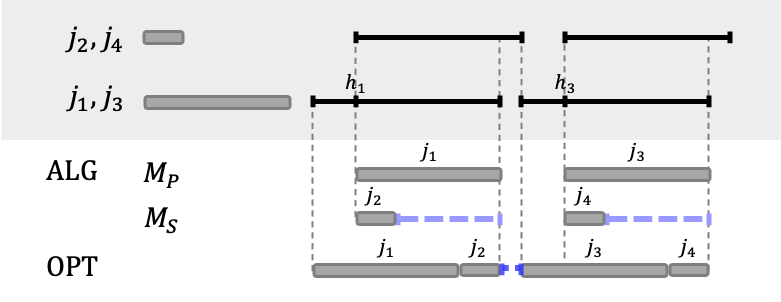}
\caption{tight example}
\label{fig:tight}
\end{figure}

Chen et al. proved Algorithm~$\textbf{S}$ is 4-competitive when $\lambda = 1$.
Here we give an example which shows that the competitive analysis is tight in~\cite{chen2015online}.
That is, the competitive ratio for the tight example is 4. 
We set $\psi_b = \psi_\sigma = 1$, and $E = \mathcal{B}=k$,
where $k$ is a sufficiently large value.
Since it has been proved that the competitive ratio for Algorithm $\textbf{S}$ is 4 when $\lambda = 1$,
the $h_j$ for job $j$ is set to be $\max\{a_j, d_j-\mathcal{B}\}$.
Assume $c_1=c_3=\mathcal{B}$ and $c_2=c_4=\epsilon$.
Algorithm $\textbf{S}$ turns on $M_1$ at $h_1$.
Since $j_2$ arrives at the same time, i.e. $a_2 = h_1$,
obviously,
$M_1$ cannot finish $j_2$ before its deadline.
Therefore, $M_2$ is switched on to execute $j_2$.
According to the strategy of Algorithm~$\textbf{S}$,
$M_2$ remains idle until processor $M_1$ keeps ON for $\mathcal{B}$.
Thus, $M_1$ and $M_2$ are switched off at the same time.
For jobs $j_3$ and $j_4$, we design the same scenario to consider the performance.
As a result, the total energy cost is $4E+4\mathcal{B}-2\epsilon=8k-2\epsilon$.
By contrast,
in the offline optimal solution, 
it turns on $M_1$ at $a_1$ and keep $M_1$ ON until all the input jobs are finished.
The energy cost is $E+2\mathcal{B}+2\epsilon+\epsilon'=3k+2\epsilon+\epsilon'$.

We can see that, for any $n$, if we let $r=n+1$, the energy cost for Algorithm~$\textbf{S}$ becomes $(4k-\epsilon)\cdot r$,
and the energy cost for the offline optimal schedule is $k + r(k+\epsilon)+(r-1)\cdot \epsilon'$.
Let the values of $\epsilon$ and $\epsilon'$ be sufficiently small so that
the two parameters can be ignored. The ratio becomes
\begin{align}
\frac{4kr}{k+kr}=\frac{k\cdot 4r}{k\cdot (r+1)}=\frac{4r}{r+1}
\end{align}
Let $r$ be a number that is sufficiently large, then the ratio is close to 4.
Thus, we have proved that the analysis of Algorithm~$\textbf{S}$ is tight.}

\longdelete{
\subsection{Easy Algorithm}

As described in Sec.~\ref{intro}, 
the key idea to escape from the malicious adversary is to delay incoming jobs as much as possible 
and to use the secondary machine for a job that cannot be scheduled because of this delay.
Let $A_0$ be a naive implementation of this idea.
Let $t$ denote the current time. 
Suppose both machines are OFF initially and requests, $j_1, j_2, \ldots$, are coming. 
Once they go to $Q$, $A_0$ delays their execution as far as possible. 
If some job, $j=(a_j,d_j,c_j)$
or maybe some set of jobs, 
reaches its due time, i.e., $d_j-t=c_j$ holds, 
then $A_0$ moves $j$ and all the other jobs in $Q$ to $Q_{M_P}$ even though some of them have more room for delay. 
Next, 
all of them are scheduled by EDF and executed continuously on $M_P$. 
If~$Q_{M_P}$ becomes empty, 
$M_P$ becomes idle for some designated time and is turned off when this idle time expires. 
Now the system returns to the initial state; 
namely both machines are OFF.

Look at Fig.~\ref{fig:up1}.
Here, we use a solid line segment from time $a_j$ to $d_j$ for each job $j=(a_j,d_j,c_j)$,
and a gray box to represent its execution time $c_j$. 
Below those line segments and boxes, 
we illustrate how those jobs are executed by $A_0$ and OPT. 
Dashed line segments show the idle time of $A_0$ and dotted line segments the idle time of OPT (if any).
We will come back to this figure later, so at this moment,
see only $j_1$ through $j_{10}$, $j_{11}$ and $j_{12}$, 
and ignore all the other inputs that will be later discussed. 
$j_1$ through $j_{10}$ have a late deadline and once enter the global queue $Q$. 
$j_{11}$ is also delayed, 
but its due time soon comes and is executed on the primary machine $M_P$ by turning it on.
Then $j_1$ through $j_{10}$ in $Q$ are also executed continuously
after $j_{11}$ on $M_P$.
(This seems natural in the sense that those jobs can be executed with absolutely minimum cost).

\begin{figure}[htb]
\centering
\includegraphics[scale=0.65]{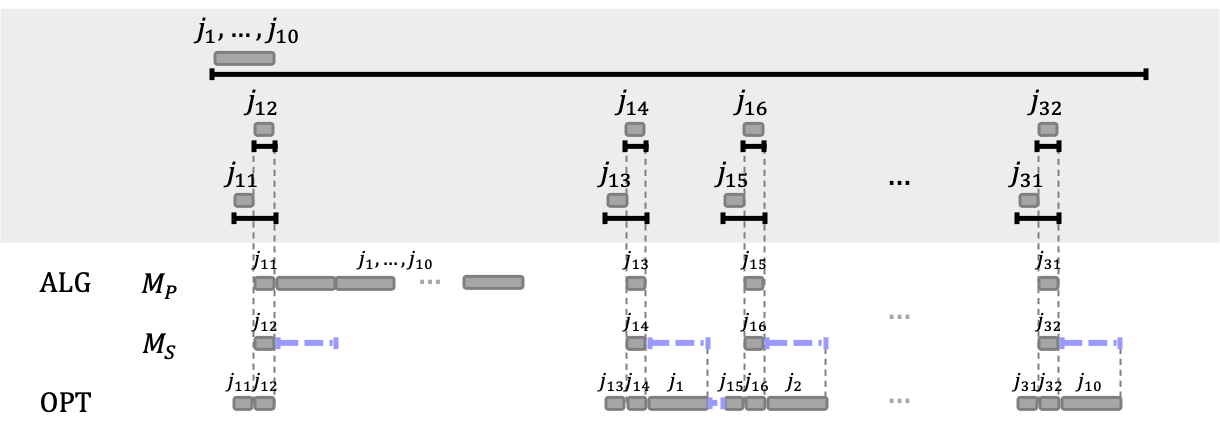}
\caption{Example 1. Note that the assignment of $M_P$ and $M_S$ is original. They are swapped whenever $M_S$ is newly turned on. That is, in this illustration, the machine having idle time turns out to be $M_P$ at that time. The amount of idle time is one.}
\label{fig:up1}
\end{figure}

During the above execution, a new job, $j'$, may arrive and there are two cases.  
One is that the current schedule for $Q_{M_P}$ can accommodate its execution. 
If so, $j'$ is inserted into $Q_{M_P}$,
which is rescheduled by EDF (we say $M_P$ is available).
The other case is that there is no room for $j'$ in $Q_{M_P}$
(if jobs in $Q_{M_P}$ would not have been delayed, 
this cannot happen because of the single machine schedulability
assumption).
We call such a $j'$ {\em urgent}.
In Fig.~\ref{fig:up1}, $j_{12}$ is urgent
(but if $j_{12}$'s deadline would be a bit later, 
it could be inserted after $j_{11}$ on $M_P$).
If an urgent job comes, then the secondary machine, $M_S$, is turned on, 
on which $j'$ is executed immediately ($M_P$ continues executing the already assigned jobs). 
Notice that when $j_1$ through $j_{12}$ have been assigned and their execution has started on $M_P$ and $M_S$, 
there are no pending jobs at all. 
Namely we can consider a new game has started. 
Furthermore, $M_S$ can act as a new primary machine whose first job is $j_{12}$ 
($j_{12}$ and all the succeeding jobs obviously satisfy Condition~\ref{input} of Lemma~\ref{assmpution_chen}). 
Therefore, we swap $M_P$ and $M_S$ at this point of time although we need a caution since the new $M_S$ may be busy for a while 
(details will be discussed later).
In the figure, we show the original $M_P$ and $M_S$, but this swap occurs whenever $M_S$ changes from OFF to ON.

What about the CR of $A_0$?
For simplicity, we assume $\ps=\psi_b$ ($=1$) for a while.
See Fig.~\ref{fig:up2-4} (a).\ccc
Let $j_1$ be a tiny job; i.e., $c_1$ is very small. 
$j_1$ is delayed and another tiny job, 
$j_2$ having the same deadline, 
comes after the execution of $j_1$ has been started 
(we will say simply ``$j_1$ is started''). 
Obviously $M_P$ is not available for $j_2$ and $M_S$ is turned on. 
Recall that $M_P$ and $M_S$ are swapped. 
After finishing $j_2$, 
$M_P$ enters an idle state to prepare for the next job that may come soon and 
suppose the amount of this idle state would be 1.
It should be noted that this 1 ($=\mathcal{B}=E/\ps$,
recall we are temporarily assuming $\ps=1$)
is a popular setting for the idle time that has been proved optimal in similar situations including the ski-rental problem.
$A_0$ spends two turn-on costs and one idle cost,
namely 2+1=3 ignoring the tiny execution costs. 
OPT spends only one turn-on cost, which is 1. 
So the CR is 3. So far so good.


\begin{figure}[htb]
\centering
\includegraphics[scale=0.55]{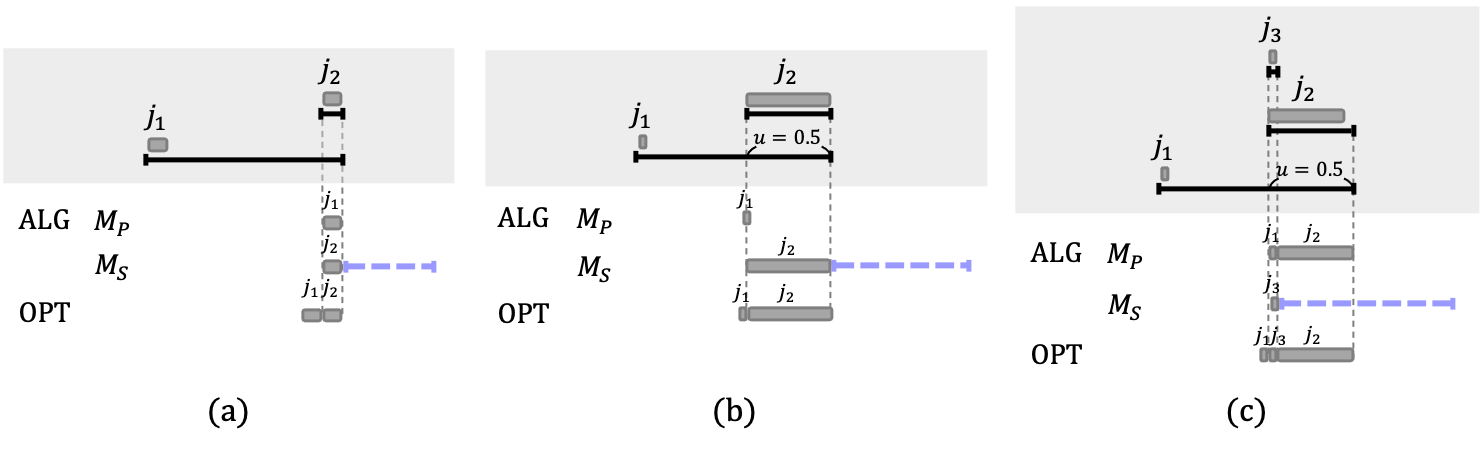}
\caption{(a) Example 2, where the idle time of new $M_P$ is one ($M_P$ and $M_S$ are swapped at~$a_2$); (b) Example 3, where the idle time of new $M_P$ is set to two ($M_P$ and $M_S$ are swapped at~$a_2$); (c)~Example 4, where the idle time of new $M_P$ is two ($M_P$ and $M_S$ are swapped at~$a_3$)}
\label{fig:up2-4}
\end{figure}

Unfortunately we have a problematic case (see Fig.~\ref{fig:up1} again).
Recall that all $j_1$ through $j_{10}$ are executed together with
$j_{11}$ on $M_P$.
However, OPT does not do so now, which is possible since those jobs have a late deadline.
Subsequently
we have another ten pairs of two jobs, 
i.e., $j_{13}$, $j_{14}$ and so on,
each of which is similar to the pair of $j_{11}$ and $j_{12}$.
Exactly as before, we need to turn on both
machines for each of the ten pairs. 
(Also we need to swap $M_P$ and $M_S$ each time, but this is not important here.)
Note that every consecutive two sets have a tiny gap, 
so we need to turn on $M_P$ each time. 
OPT simply keeps being ON during this period since it is better (not worse at least) than inserting turn-off and turn-on in the middle. 
The total cost of $A_0$ is $1+1+1+10+10\times(1+1+1)=43$
(where three 1's are two turn-on costs and one idle cost),
whereas the cost of OPT is $1+1+10\times1=12$, 
by ignoring a tiny idle time between every two pairs of jobs. 
The trick is that OPT can execute the pending $j_1$ through $j_{10}$ using the time slots that $A_0$ is in an idle state without any additional cost.
Thus the CR is much worse than 3 as calculated above. 
Actually, we can increase the number of such pairs; 
one can easily see the CR approaches to 4.
Now we fix this problem.}

Algorithm~\oural\ has two key gadgets. 
First, we present 
a new notion of ``\emph{margin}'',
as a positive constant $u$, 
for delaying a job. 
That is, 
we start a job earlier than its due time by~$u$. 
Fig.~\ref{fig:margin}
shows examples to illustrate the idea. 
Here, we use a solid line segment from time $a_j$ to $d_j$ for each job $j=(a_j,d_j,c_j)$,
and a gray box to represent its execution time $c_j$. 
Below those line segments and boxes, 
we illustrate how those jobs are executed by ALG and OPT. 
Dashed line segments show the idle time of ALG and dotted line segments the idle time of OPT (if any).

\begin{figure}[htb]
\centering
\includegraphics[scale=0.65]{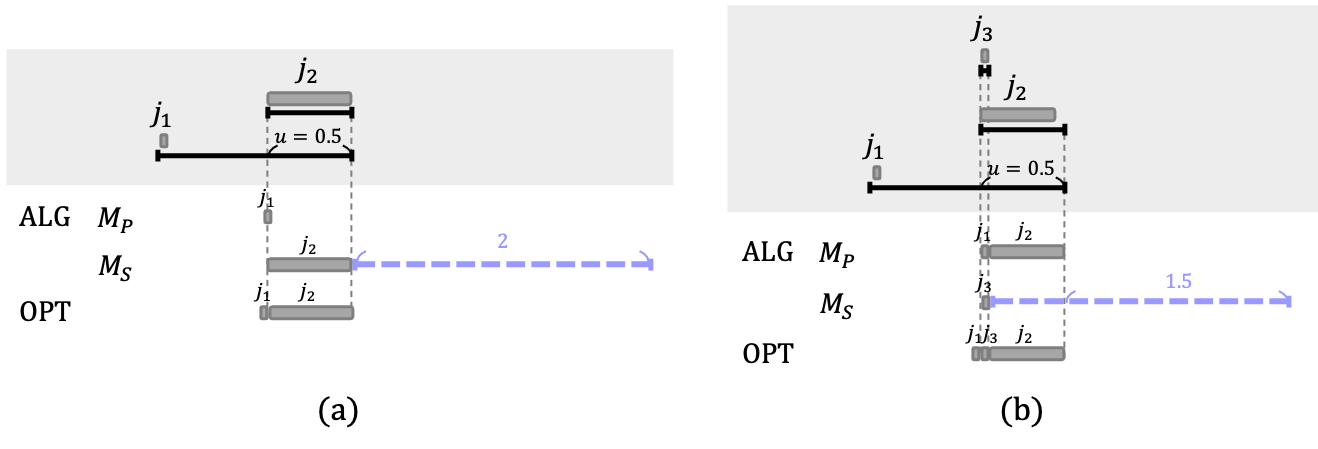}
\caption{Two examples for illustrating the notion of margin}
\label{fig:margin}
\end{figure}

As shown in Fig.~\ref{fig:margin} (a), 
$j_1$,
having a tiny execution time of $\epsilon$,
starts at $d_1-u -\epsilon$.
Here we let $u=0.5$ to get close to the optimal setting
although its perfect setting seems hard.
(We will show more details about the margin setting in the next section.)
Thus, it is necessary that job $j_2$ that provokes $M_S$'s turn-on has an execution time of at least 0.5 
(otherwise it can interrupt $j_1$ and be inserted into $M_P$).
See Fig.~\ref{fig:margin} (b) for a more complicated scenario.
Now $j_2$ is a bit smaller than the one in (a) and can be inserted into the (old) $M_P$'s execution gap caused by the margin $u$.
Then a tiny $j_3$, which is urgent, makes the secondary machine ON,
but its idle time can start 0.5 earlier than before.
Note that the idle time of a machine is typically set to $\mathcal{B}$ after the last job is finished,
where $\mathcal{B}=E/\ps=1$ since we are temporarily assuming $\ps=1$.
It is actually a popular setting for the idle time that has been proved optimal in similar situations including the ski-rental problem.
In contrast, 
the second key gadget of Algorithm~\oural\ is that we set this value to twice, giving rise to that OPT has an interval of 1.5 from the end of $j_2$ (i.e. $d_2$) to the end of the idle state of $M_P$. 
On the other hand,
recall that OPT has an interval of 2 in the previous scenario 
(see Fig.~\ref{fig:margin} (a)).  
One can observe that OPT could use this interval to execute the pending jobs and thus the longer the better for this interval.
However, we have to use 1.5 instead of 2 for the analysis in Section~\ref{analysis}.

\longdelete{
Now the cost of \oural\ is $1+1+0.5+x$ 
(including two turn-on costs and an idle cost of $x$) 
and that of OPT $1+0.5$ (including one turn-on cost).
We can thus increase the idle time $x$ from 1 to 2 to achieve the same CR of 3 for this small example. %
Why does this increased idle time help?
Because we need only five pairs of jobs $j_i$ and $j_{i+1}$ to realize the same situation as shown in Fig.~\ref{fig:up1};
the turn-on cost of ALG halves while the cost of OPT seldom changes.

\begin{figure}[htb]
\centering
\includegraphics[scale=0.7]{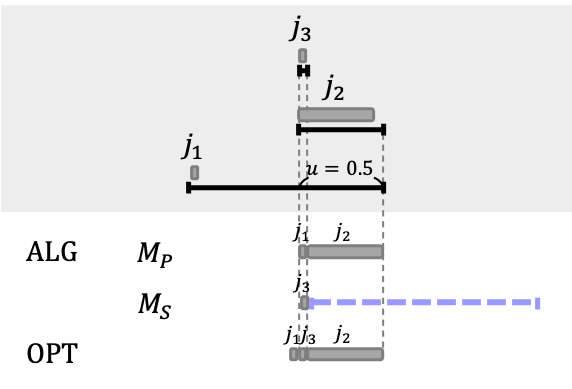}
\caption{Example 4, where the idle time of new $M_P$ is two ($M_P$ and $M_S$ are swapped at~$a_{j_2}=a_{j_3}$)}
\label{fig:up4}
\end{figure}}

\longdelete{
In summary, there are two scenarios to be considered.
One is that new $M_P$ finishes its execution later than new $M_S$ (Fig.~\ref{fig:margin} (a)), and the other is that new $M_P$ finishes the execution first since the job that provokes new $M_P$ ON is tiny, resulting in the idle time of new $M_P$ is actually overlapped with the execution of the new $M_S$ (Fig.~\ref{fig:margin} (b)).}

\medskip
{\bf Tight Analysis of Algorithm $\textbf{S}$~\cite{chen2015online}.}
The two gadgets of Algorithm~\oural\ reveal the improvement of the upper bound in some sense. 
Here we recall  
Chen et al.'s 4-competitive algorithm, 
called \emph{Algorithm $\textbf{S}$}, 
and conduct tight analysis of their algorithm 
to clarify the merit of our proposed gadgets.  
That is, we present a tight example to show its worst competitive performance.

\begin{figure}[htb]
\centering
\includegraphics[scale=0.65]{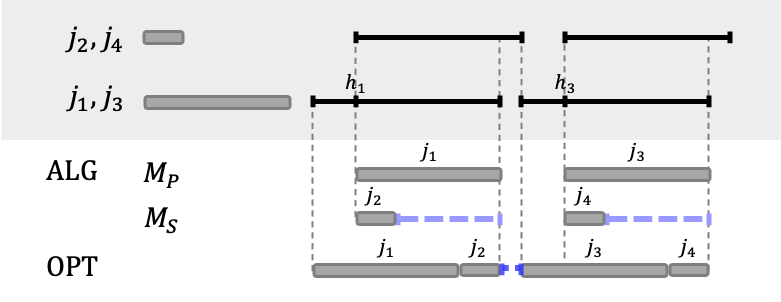}
\caption{Tight example for Algorithm $\textbf{S}$~\cite{chen2015online}}
\label{fig:tight}
\end{figure}

\longdelete{
Based on the previous assumption that a single machine can finish all the input jobs in the offline model,
Algorithm $\textbf{S}$ basically uses one machine as far as possible and uses the second process if necessary.
That is, they set $M_1$ to be the main machine, 
and $M_2$ to be the backup machine.
When a job arrives,
Algorithm $\textbf{S}$ always turn on the main machine $M_1$ first and decides to turn on the backup machine $M_2$ only if it finds out that $M_1$ is not able to finish all the pending jobs before their deadlines;
i.e., $M_1$ is overloaded.

More precisely,}

Basically, 
the main structure of Algorithm $\textbf{S}$ is similar, 
but they let every job 
$j$ be associated with a parameter $h_j = \max\{a_j, d_j- \lambda\mathcal{B}\}$,
i.e. its \emph{energy-efficient anchor}, 
to determine when the main machine $M_P$ should be switched on for a pending job $j$, 
where $\lambda$ is a constant (setting to be one in~\cite{chen2015online}).
Set $\psi_b = \psi_\sigma = 1$, and $E = \mathcal{B}=k$,
where $k$ is a sufficiently large value, for simplicity.
As shown in Fig.~\ref{fig:tight}, we let 
$c_1=c_3=\mathcal{B}$ and $c_2=c_4=\epsilon$.
Algorithm $\textbf{S}$ turns on $M_P$ at $h_1$.
Since $j_2$ arrives at the same time, i.e. $a_2 = h_1$,
obviously,
$M_P$ cannot finish $j_2$ before its deadline.
Therefore, $M_S$ is switched on to execute $j_2$.
One can observe that the design of the margin in Algorithm \oural\ can provide the flexibility to escape from the scenario.

\longdelete{
The key steps of Algorithm $\textbf{S}$ are summarized as follows.
When a new job $j$ arrives and it turns out that $M_1$ cannot finish all the pending jobs before their deadlines if $j$ is assigned to $M_1$,
then switch the state of $M_2$ to ON.
Let the current time be $t^*$ and schedule the jobs that arrived before $t^*$ on $M_1$ based on the earliest-deadline-first (EDF) principle,
and assign the jobs that arrive after $t^*$ to $M_2$.
If both processors $M_1$ and $M_2$ are ON, then Algorithm $\textbf{S}$ decides to turn off $M_1$ at the moment when it becomes idle.
If only one of $M_1$ and $M_2$ is ON, say $M_2$, then when $M_2$ becomes idle, it is switched off once the length between the current time and the time when $M_1$ is turned on has reached $\mathcal{B}$.

Chen et al. proved Algorithm~$\textbf{S}$ is 4-competitive.
when $\lambda = 1$.
Here we give an example which shows that the competitive analysis is tight in~\cite{chen2015online}.
That is, the competitive ratio for the tight example is 4.}

Next, 
the turn-off strategy of Algorithm~$\textbf{S}$ is as follows: 
if both machines are ON, $M_P$ (new $M_S$) is turned off 
as soon as it becomes idle.
Otherwise, if only one of the machines is ON, 
then when the machine becomes idle, it is switched off once the length between the current time and the moment when $M_P$ was turned on has reached $\mathcal{B}$.
In this example,
$M_S$ remains idle until 
$M_P$ keeps ON for $\mathcal{B}$ units of time,
resulting in $M_P$ and $M_S$ being switched off at the same time.
One can observe that 
if the job executed on $M_S$ is tiny, 
the malicious adversary can punish the turn-off strategy 
since it turns off $M_S$ too early. 
We design the same worst-scenario
for jobs $j_3$ and $j_4$ and consider the performance.
As a result, the total energy cost of ALG is $4E+4\mathcal{B}-2\epsilon=(4k-\epsilon)\cdot 2$.
By contrast,
in the offline optimal solution, 
it turns on a single machine at $a_1$ and keep it ON until all the input jobs are finished.
The minimum energy cost is $E+2\mathcal{B}+2\epsilon+\epsilon'=k + 2(k+\epsilon)+\epsilon'$.

It is not hard to see that, 
if we consider the above case of four jobs to be one round, 
and the energy cost of Algorithm~$\textbf{S}$ for $\frac{r}{2}$ rounds, 
where $r$ is a multiple of 2,
is $(4k-\epsilon)\cdot r$,
while the energy cost of OPT is 
$k + r(k+\epsilon)+(r-1)\cdot \epsilon'$.
Letting the values of $\epsilon$ and $\epsilon'$ be very small, 
the ratio becomes
\begin{align*}
\frac{4kr}{k+kr}=\frac{k\cdot 4r}{k\cdot (r+1)}=\frac{4r}{r+1}.
\end{align*}
When the value of $r$ is sufficiently large, 
the CR approaches 4.
The tight analysis reveals the advantages of our designed gadgets in Algorithm~\oural.

\longdelete{
\begin{table}[h]
\caption{Algorithm \oural}
\vspace{-18pt}
\begin{center}
\begin{tabular}{l}
\hline
\noalign{\smallskip}
    Initially assign\ccc $M_P$ and $M_S$ to two machines arbitrarily.\\
    $u$ is the margin and for $W$, $Q$, $\ps$ and so on, see the
    previous part of this section.\\ 
    The input satisfies Condition~\ref{input} of Lemma~\ref{assmpution_chen}.\\
    At any time $t$, \oural\ proceeds as follows:\\
\noalign{\smallskip}
\hline
\noalign{\smallskip}
\begin{minipage}{5.5in}
    \vskip 2pt
    \begin{enumerate}
      \item Execution of jobs:
      \begin{enumerate}
        \item No machines are on.\\ 
        If there exist $t^{\dagger}$ and $t^{\ast}$ such that $W(t^{\dagger},t^{\ast})\geq t^{\ast}-t^{\dagger}$ and $t\geq t^{\dagger}-u$, turn on one machine ($M_P$) and move all jobs in $Q$ to $Q_{M_P}$.
        \item Only $M_P$ is on and a new job $j$ comes.\\
        If there is no $t^{\ast}$ such that $W(t,t^{\ast})>
        t^{\ast}-t$ for jobs in $Q_{M_P} \cup \{j\}$ (i.e., $M_P$ is available), add $j$ to $Q_{M_P}$ and reschedule it.\\
        Otherwise, turn on the other machine ($M_S$) and add $j$
        to $Q_{M_S}$ that is empty.\\
        After that switch $M_P$ and $M_S$.
	    \item $M_S$ is on and a new job $j$ comes.\\
    	If $M_S$ is available, add $j$ into $Q_{M_S}$ and reschedule it.\\
        Otherwise move it to $Q_{M_P}$ (it is guaranteed that $M_P$ is available).
      \end{enumerate}
      \item Idle state:
      \begin{enumerate}
        \item $M_S$ never has an idle state.\\
        That is, we immediately turn off $M_S$ once $Q_{M_S}$ becomes empty.
        \item $M_P$ becomes idle once $Q_{M_P}$ becomes empty.\\ 
        We turn off $M_P$ when its total idle time becomes $2/\ps$.\\
      \end{enumerate}
    \end{enumerate}
    \vskip 2pt
\end{minipage}
\vspace{-5pt}
\\
\hline
\end{tabular}
\end{center}
\end{table}
}

\section{Analysis of the Algorithm}
\label{analysis}
For analysis, we need to show that (i) \oural\ is correct, 
namely it can execute any sequence of jobs arriving under the single machine schedulability condition 
and (ii) its CR is at most 3.
We mostly focus on (ii).
(i) is not hard and the following observation should be enough to see
every job is executed:

(1) Suppose when a new job $j$ comes at time $t$, no machine is
ON. Then if the condition of 1-a for $Q\cup {j}$ is not met, $j$ just
goes to queue $Q$. Once $j$ enters $Q$, it must be executed eventually
by 1-a. If the condition of 1-a is met, then all the jobs in $Q$
go to $Q_{M_P}$ and we go to 1-b. Here $j$ is inserted to $Q_{M_P}$ (if
possible) or is executed on $M_S$ that is turned on at $t$. Thus $j$ is
executed. 

(2) Suppose only $M_P$ is already ON but $M_S$ is OFF when $j$
comes. As (1), $j$ is inserted to $Q_{M_P}$ or executed on $M_S$ that
is turned on at $t$.

(3) Suppose $M_S$ is ON when $j$ comes. If $M_S$ is available,
$j$ is executed on it. Otherwise, note that when $M_S$ is turned on
for the last time, all the jobs in $Q$ go to $M_S$ (recall $M_P$ and
$M_S$ are swapped) and while $M_S$ is on,
all newly coming jobs are executed on $M_P$ due to EDF. This execution
is possible since if the whole input sequence satisfies Condition~\ref{input} of Lemma~\ref{assmpution_chen}, 
its arbitrary suffix obviously does, too.
$j$ is one of them and must be executed. 

To prove the CR, we first define a {\em phase}. 
Suppose the system changes its state from both machines OFF to
at least one machine ON at $t_1$ and returns to the state of both
machines OFF at $t_e$. \ccc
Also let $a_1, a_2, \ldots, a_k$ be the arrival time of the jobs executed during $t_1$ to $t_e$ and let $t_0$ be $\min\{a_1, \ldots, a_k\}$. 
Then the time slot from $t_0$ to $t_e$ is called a phase.
Note that an entire execution of \oural\ consists of some phases.
Let $P(i)$ be the $i$'th phase.

\begin{figure}[htb]
\centering
\includegraphics[scale=0.6]{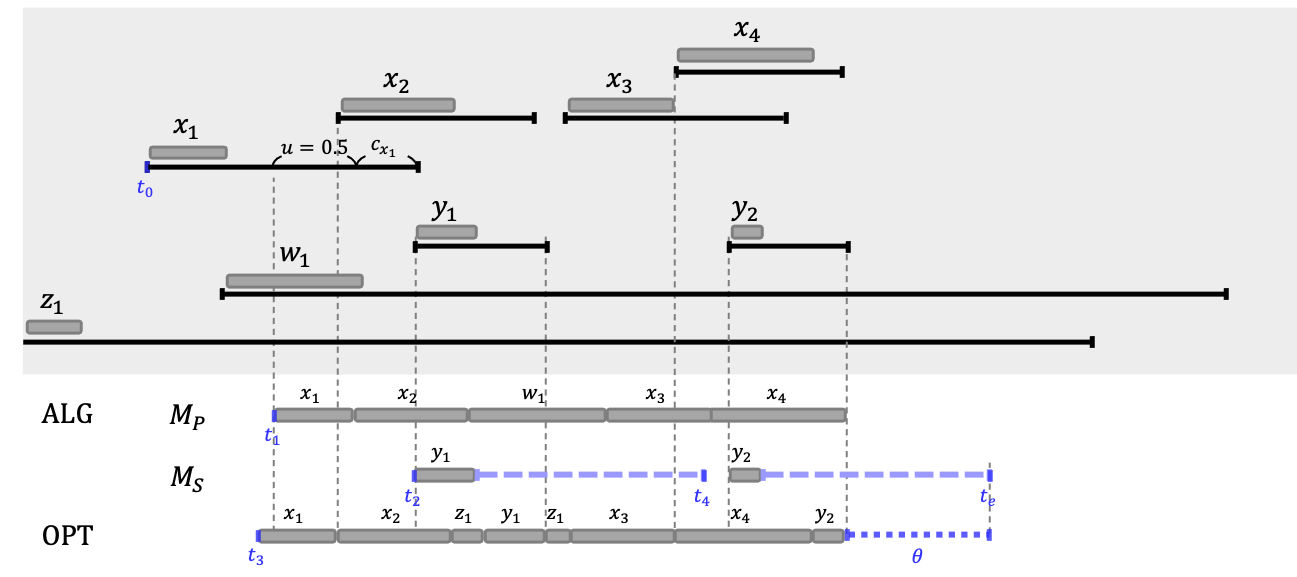}
\caption{A single phase. Now a period of idle time is $2/\ps$. $M_P$ and $M_S$ are swapped at~$a_{y_1}$.}
\label{fig:up5}
\end{figure}

Fig.~\ref{fig:up5} illustrates how a single phase looks like.
$x_1$ and $x_2$ are executed on $M_P$, 
then an urgent $y_1$ comes and executed on the new $M_P$ 
(so $w_1$, $x_3$ and $x_4$ are executed on the new $M_S$).
$w_1$ is a job moved from $Q$ to $Q_{M_P}$ when $x_1$ starts
(more precisely, $w_1$ is scheduled after $x_1$, 
but $x_2$ is inserted later due to the EDF principle).
$x_3$ and $x_4$ are added to $Q_{M_S}$. 
$y_2$ is a job that cannot be inserted to $Q_{M_S}$ because of the
existence of $x_4$; it goes to $M_P$.
Thus $M_P$ can be OFF for some period of time during a single phase 
(but $M_S$ should be ON during that period by definition). 
The earliest arrival time of the jobs \oural\ executes in this
phase is that of $x_1$, which is the start time, $t_0$, of this
phase.
Note that $t_e$ at which the phase ends is the moment when the idle time after $y_2$ expires.
An example of the OPT's execution sequence is given at the bottom of the figure. 
It must execute $x_1$, $x_2$, $x_3$, $x_4$, $y_1$ and $y_2$ since each of them has its arrival time and deadline within the phase.
$w_1$ is not executed by OPT in this example.
By contrast, $z_1$ which is performed by OPT is not executed by \oural\ (can be executed in the previous phase),
so its arrival time is not counted to determine the beginning of the phase.
(Note that OPT and \oural\ as well, may execute each job separately, so
we have two $z_1$'s in the figure.)
Fig.~\ref{fig:up5} is also used for the later analysis of the competitive ratio. 
A phase is called a single-machine phase if only one machine ($M_P$) is ON in that phase and a dual-machine phase otherwise.

Our proof of the CR uses a math induction.
Since \oural\ is deterministic, the set of jobs that are executed by \oural\ in $P(i)$ is uniquely determined once the input is given, 
which we denote by $J_A(i)$.  
For the jobs executed by OPT in $P(i)$, the situation is less clear; 
jobs whose arrival time and deadline are both within $P(i)$ must be executed, 
but jobs such that only one of them is within $P(i)$ may or may not be executed even partially. 
Furthermore, OPT may execute some jobs outside phases, 
namely while no machine is ON. 
Fix an arbitrary execution sequence $S(i)$ of OPT in this phase.
Then we define the following parameters used in the induction.

(1) $\alpha(i)$ ($\lambda(i)$, and $\delta(i)$, resp.) $=$ the total execution time of the jobs executed by both \oural\ and OPT (only by \oural\, and only by OPT possibly partially, resp.) 

(2) $A(i)$ is the cost of \oural\ in $P(i)$ that includes
$\alpha(i)$, $\lambda(i)$, turn-on costs and idle costs. 

(3) $O_f(i)$ is a lower bound for the cost of OPT, 
namely it includes $\alpha(i)$, $\delta(i)$ and idle costs if any. 
Note that it does not include the turn-on cost when $S(i)$ starts since OPT may not need it depending on its state at the end of the previous phase, 
but does include one(s) if $S(i)$ includes turn-off and turn-on in the middle of the phase. 
$O_n(i)$ is similar but we impose the condition that OPT is ON at the end of $P(i)$. 

For instance, consider $P(i)$ whose execution sequence looks like Fig.~\ref{fig:margin} (a).
Then $A(i)=4.5$, $O_f(i)=0.5$ and $O_n(i)=2.5$,
where 0.5 is for the execution of jobs and $2=(2/\ps)\ps$ is the idle cost to keep it ON until the end of the phase.\ccc

\begin{theorem}
\oural\ is correct and its CR is at most 3 for $u=0.5$.\ccc
\label{ubtheorem}
\end{theorem}

\begin{proof}
We omit the first part (see the previous observation).
For the CR, we fix an arbitrary execution sequence $S$ ($S(i)$ is
its subsequence associated with $P(i)$)
for the entire execution sequence of OPT and prove two lemmas.

\begin{lemma}\label{lemma3}
The following (2) and (3) hold for each phase for $r=3$.
\begin{align}
&rO_f(i)-A(i)\geq\delta(i)-\lambda(i)-r.\\
&rO_n(i)-A(i)\geq\delta(i)-\lambda(i).
\end{align}
\end{lemma}

The proof will be given later. 
Let $O(i)$ be the (real) cost of OPT under the sequence $S(i)$ in $P(i)$. 
Also let $m$ be an integer less than or equal to the number of phases.

\begin{lemma}
For $r=3$, we have
\begin{numcases}{\sum_{i=1}^{m}\left(rO(i)-A(i)\right)\geq}
\sum_{i=1}^{m}\left(\delta(i)-\lambda(i)\right)
\text{\hspace{3mm} if OPT is OFF at the end of $P(m)$.}\\
\sum_{i=1}^{m}\left(\delta(i)-\lambda(i)\right)+r
\text{\hspace{3mm} if OPT is ON at the end of $P(m)$.}
\end{numcases}
\end{lemma}
\begin{proof}
Suppose OPT is OFF at the end of $P(i)$. 
Then $O(i)$ should be at least $O_f(i)$ since the latter is a lower bound and similarly for $O_n(i)$ if OPT is ON at the end of the phase.
For $m=1$, note that OPT must spend the turn-on cost of 1 that is not included in either $O_f(i)$ or $O_n(i)$.
Therefore if OPT is OFF at the end of the phase, we have
\begin{align*}
&rO(1)-A(1)\geq rO_f(1)+r-A(1) \geq\delta(1)-\lambda(1)
\end{align*}
by Lemma 3.
Otherwise, if OPT is ON at the end of the phase, we have
\begin{align*}
&rO(1)-A(1)\geq rO_n(1)+r-A(1) \geq\delta(1)-\lambda(1)+r
\end{align*}
similarly. 
Now suppose the lemma is true for $m'=m-1$. 
Then to prove that the lemma also holds for $m'=m$, 
we consider four cases and define $C \rightarrow D$ as the states of OPT at the end of $P(m-1)$ and $P(m)$ respectively:
(i) OFF $\rightarrow$ OFF, (ii) OFF $\rightarrow$ ON, (iii) ON $\rightarrow$ OFF and (iv) ON $\rightarrow$ ON.
For case (i), OPT must pay the turn-on cost in $P(m)$, so
\begin{align*}
\sum_{i=1}^{m}\left(rO(i)-A(i)\right)&=\sum_{i=1}^{m-1}\left(rO(i)-A(i)\right)
+rO(m)-A(m)\\
&\geq\sum_{i=1}^{m-1}\left(\delta(i)-\lambda(i)\right) +0\cdot r
+rO_f(m)-A(m)+1\cdot r\\
&\geq\sum_{i=1}^{m-1}\left(\delta(i)-\lambda(i)\right)+\delta(m)-\lambda(m)\\
&=\sum_{i=1}^{m}\left(\delta(i)-\lambda(i)\right)
\end{align*}
by the induction hypothesis and Lemma 3.
Here 0 and 1 before $r$ are for handling the four cases. 
Since the current case is OFF $\rightarrow$ OFF, the first 0 means formula (4) does not have $r$ on the right-hand side and the second 1 means that OPT must turn on in $P(m)$.
The other cases are similar (just 0 and 1 before $r$ change) and may be omitted.
\end{proof}

If there are $m$ phases in total, it must be that
$\sum_{i=1}^{m}\left(\delta(i)-\lambda(i)\right)=0$ 
and thus the theorem is proved. 
What remains is to prove Lemma 3.

\begin{proof}
(Proof of Lemma 3)
Suppose $x_1=(a_1,d_1,c_1)$ is the first job executed in some phase $P$ starting from $t_0$ and ending at $t_e$.
Then since $x_1$ is executed in $P$,
$t_0\leq a_1$ (there may be another job executed in $P$ and having an earlier arrival time).
Also, since the ending time of $x_1$'s execution is to be $d_1-0.5$
due to the delay
(or not delayed at all if this amount of delay is impossible)
and we have a mandatory idle time, $2/\ps\geq2$, after its (or a later job's) execution, 
it must be that $d_1\leq t_k$. 
This means the period ($a_1,d_1$) is included in $P$, meaning OPT also executes $x_1$ in $P$. 
Thus, in each phase, both \oural\ and OPT execute at least one job. 
Also it turns out, by definition, that \oural\ never executes a job outside phases.

Note, however, that OPT may execute some job, say $x$, outside phases
(this happens, e.g., if $u_1$ in Fig.~\ref{fig:up6} is executed by OPT after this phase and before the next phase). 
If that happens, we consider that $x$ is executed in a ``special'' phase, $P(i')$,
by extending the definition of a phase. 
Note that this phase has $A(i)=\lambda(i)=0$ and both $O_f(i)$ and $O_n(i)$ are at least $\delta(i)$, 
so (2) and (3) obviously hold. 
A special phase may continue to/from a neighboring (normal) phase.

\begin{figure}[htb]
\centering
\includegraphics[scale=0.6]{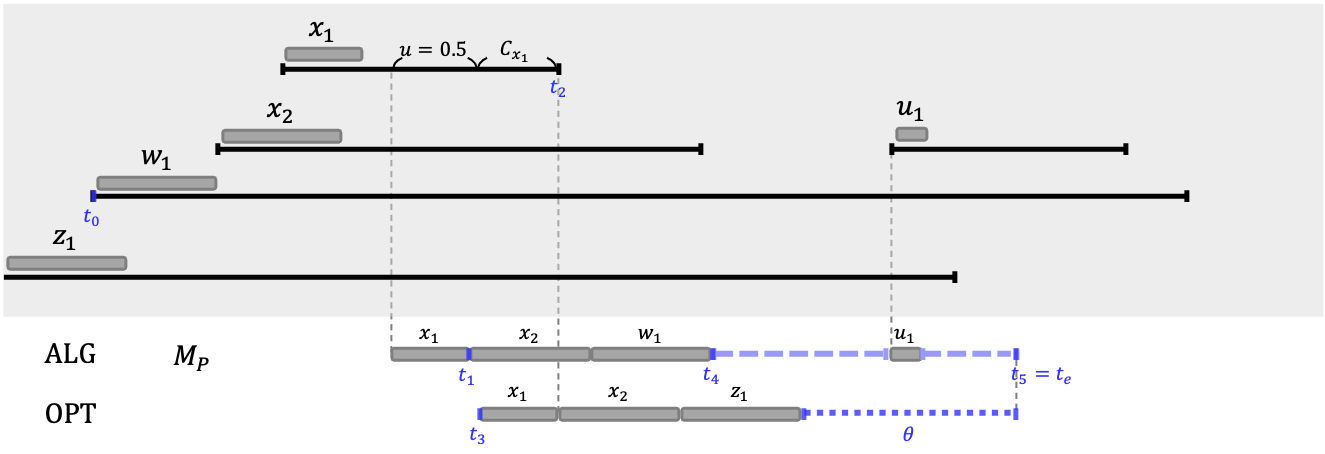}
\caption{A single-machine phase. Recall a period of idle time of \oural\ is $2/\ps$.}
\label{fig:up6}
\end{figure}

We first prove the lemma for a single-machine phase $P(i)$.
See Fig.~\ref{fig:up6} which illustrates the execution sequences of \oural\ and OPT as with Fig.~\ref{fig:up5}.
In this figure,
the phase starts at the arrival time of $w_1$ and ends when ALG's
machine ($M_P$) turns off.
$x_1$ is the first job executed in this phase,
and $x_2$ is a job whose arrival time and deadline are both within this phase 
(so must be executed by both \oural\ and OPT in this phase). 
$z_1$ has only its deadline, and
$u_1$ and $w_1$ have only their arrival times within the phase.
Thus there are several different types of jobs, 
but what is important is whether or not each job is executed in this phase.
In the example of this figure,
$x_1$ and $x_2$, called {\em type-AO} jobs, are executed by both \oural\ and OPT, 
$w_1$ and $u_1$, {\em type-A}, by only \oural, 
and $z_1$, {\em type-O}, only by OPT. 
Each type can include an arbitrary number of jobs, 
but as will be seen in a moment, 
those numbers are not important. 
So we will progress our analysis by using these five jobs, 
$x_1,x_2,w_1,u_1,z_1$, 
for a while and will mention the generalization after this analysis.
It should also be noted that a single job may be executed by OPT in two or more phases. 
If this happens, we divide that job into two or more parts and allow each part can have a different job type, if necessary.

Set the following moments (see the figure) of time: $t_1$ and $t_2$: 
the ending time of execution and the deadline of $x_1$, respectively.
$t_3$: the time when OPT becomes ON after time $t_0$ to execute $x_1$ or maybe another job.
$t_4$ and $t_5$: the start and ending times of the idle state of \oural, respectively. 
Thus one can observe that neither the number of jobs in each type nor their execution sequence is important to define these times,
except $x_1$ that is first executed.
For the five jobs, it turns out that $A(i)=c_{x_1}+c_{x_2}+c_{w_1}+c_{u_1}+1+2$ 
(1 is the turn-on cost and 2 is the idle cost; 
recall once \oural\ enters an idle state,
it continues until its total time becomes $2/\ps$,
i.e., until its total idle cost becomes $(2/\ps)\ps=2$),
and $O_f(i)\geq(c_{x_1}+c_{x_2}+c_{z_1})$.
OPT's execution of some job may exceed the border of the phase. 
If that happens, as mentioned before, 
we can partition that job into two parts, 
the first one ends at the end of the phase and the second one is to be the remaining part, 
which is executed in the following special phase.

Now we have
\begin{align*}
rO_f(i)-A(i)&\geq 3(c_{x_1}+c_{x_2}+c_{z_1})-(c_{x_1}+c_{x_2}+c_{w_1}+c_{u_1}+3)\\
&=2(c_{x_1}+c_{x_2}+c_{z_1})+c_{z_1}-(c_{w_1}+c_{u_1})-3\\
&\geq c_{z_1}-(c_{w_1}+c_{u_1})-3.
\end{align*}

Thus (2) holds for any nonnegative values of $c_{x_1}$ through $c_{u_1}$.
Similarly, letting $\theta$ be the idle time of OPT (if any) that keeps the OPT's machine ON until $t_5$ 
(OPT can turn off once and turn on at or before $t_5$, but OPT would then need an extra turn-on cost and our proof becomes easier), 
we have $O_n(i)\geq c_{x_1}+c_{x_2}+c_{z_1}+\ps \theta$ and hence
\begin{align*}
rO_n(i)-A(i)&\geq 3(c_{x_1}+c_{x_2}+c_{z_1}+\ps \theta)-(c_{x_1}+c_{x_2}+c_{w_1}+c_{u_1}+3)\\
&\geq 2(c_{x_1}+c_{x_2}+c_{z_1}+\ps \theta)+c_{z_1}-(c_{w_1}+c_{u_1})-3.
\end{align*}

Here we have the following claim.
\begin{clm}
$$c_{x_1}+c_{x_2}+c_{z_1}+\ps \theta \geq 1.5$$
\label{claim}
\end{clm}
\begin{proof}
We have a sequence of time relations:\\
\hspace*{15mm} (i) $t_2-t_1\leq 0.5$ due to the margin,\\
\hspace*{15mm} (ii) $t_3\leq t_2$ ($x_1$ must be executed before its deadline),\\
\hspace*{15mm} (iii) $t_1\leq t_4$ (obviously),\\
\hspace*{15mm} (iv) $t_3-t_4\leq t_3-t_1$ (by (iii)) $\leq
t_2-t_1$ (by (ii)) $\leq 0.5$ (by (i)),\\
\hspace*{15mm} (v) $t_5-t_4\geq 2/\ps$ (the total idle time),\\
\hspace*{15mm} (vi) $t_4-t_5\leq -2/\ps$ (inversion of (v)), and\\
\hspace*{15mm} (vii) $c_{x_1}+c_{x_2}+c_{z_1}+\theta=t_5-t_3=-(t_3-t_5)\geq 2/\ps -0.5$ (by (iv)+(vi)).\\
We may need a bit more explanation for (i). Recall that if $x_1$ is
delayed and nothing happens, $t_2-t_1$ is exactly 0.5. However, some
(small) job that comes after $x_1$'s execution has started
may interrupt $x_1$ and be inserted. It should also be noted that if
$d_{x_1}-(a_{x_1}+c_{x_1}) < 0.5$, then $x_1$ is not delayed in the
first place. Now we have from (vii) (note $\ps\leq 1$)\ccc
\hspace{20mm} $$c_{x_1}+c_{x_2}+c_{z_1}+\ps\theta \geq \ps(c_{x_1}+c_{x_2}+c_{z_1}+\theta)
\geq 2-0.5\ps \geq 1.5.$$
\end{proof}

Thus we have
\begin{align*}
rO_n(i)-A(i)&\geq 2\times 1.5 +c_{z_1}-(c_{w_1}+c_{u_1})-3=c_{z_1}-(c_{w_1}+c_{u_1}),
\end{align*}
meaning (3) is also true.
In general, we have more (or less) jobs in each job type. 
However,
one can see that the definitions of $t_1$ through $t_5$ are not affected by that (only the first job, $x_1$ is important). 
Also we can simply replace $c_{x_1}$ by the sum of execution times of type-AO jobs and similarly for type-A and type-O jobs. 
Thus the extension to the general case is straightforward and details may be omitted.

We next consider a dual-machine phase as shown in Fig.~\ref{fig:up5},
where\ccc there are two executions, $y_1$ and $y_2$ on $M_P$
while $M_S$ is busy. (Recall $M_P$ and $M_S$ are swapped when $M_S$
turns on. Now $M_S$ is busy and it may not be available for urgent jobs.)\ccc
Let the first execution be $E_1$ and the second one $E_2$.
We first consider the case that $E_2$ does not exist. Namely, there
is no $x_3$, $x_4$ or $y_2$ and OPT has an idle time after having
executed $z_1$.
Thus the phase would be ending at the end of the idle state following
the execution of $y_1$\ccc and we prove the lemma for this case first. 
As before, we use specific examples for jobs to be executed in this phase,
$x_1$, $x_2$ and $y_1$ for type-AO, $w_1$ for type-A and $z_1$ for
type-O. Since the $M_S$'s turn-on is provoked, there must be a set $S$
of jobs such that their execution on $M_P$ is impossible
during the period from some $t'_1$ to some $t'_2$.\ccc
We assume that $S=\{x_1,x_2,y_1\}$, where $t'_1=t_1$ and $t'_2$ is the
deadline of $y_1$.\ccc
What we do for the generalization is the same as before, 
namely we replace $\{x_1,x_2,y_1\}$ by the real jobs in $S$, 
maybe more jobs for type-AO, replace $w_1$ by real type-A jobs, and $z_1$ by real type-O jobs.
  
Now we start with definitions of time moments (see the figure) as before. $t_1$, $t_2$ and $t_3$: the time when $M_P$ and $M_S$ become on and the time OPT becomes busy, respectively. 
$t_4$: the time when the idle state of $M_P$ is ended and this is
the end of the phase, too.\ccc
For those five jobs, we have
$A(i)=c_{x_1}+c_{x_2}+c_{y_1}+c_{w_1}+2+2$ (we now need to turn on
both $M_P$ and $M_S$), $O_f(i)\geq (c_{x_1}+c_{x_2}+c_{y_1}+c_{z_1})$,
and $O_n(i)\geq (c_{x_1}+c_{x_2}+c_{y_1}+c_{z_1}+\ps \theta)$, where $\theta$ is
the idle time (if any) to keep the OPT's machine on until $t_4$.
  
To prove formulas (2), we have (recall $\ps\leq 1$)
\begin{align*}
rO_f(i)-A(i)&\geq 3(c_{x_1}+c_{x_2}+c_{y_1}+c_{z_1})-(c_{x_1}+c_{x_2}+c_{y_1}+c_{w_1}+4.0)\\
&=2(c_{x_1}+c_{x_2}+c_{y_1}+c_{z_1})+c_{z_1}-c_{w_1}-4.0.\\
\end{align*}
Here we can claim that $c_{x_1}+c_{x_2}+c_{y_1}\geq 0.5$. 
The reason is that there is a margin of 0.5 between the end of the
execution of $x_1$ and its deadline (recall again if realizing this
margin is impossible, $M_S$ would not have been turned on).\ccc
So if $c_{x_1}+c_{x_2}+c_{y_1}< 0.5$, 
then it follows of course $c_{x_2}+c_{y_1}< 0.5$, 
which means that $x_2$ and $y_1$ could have been executed using this margin time 
(they can interrupt the execution of $x_1$) on $M_P$, resulting in a contradiction. 
Thus $rO_f(i)-A(i)\geq 2\times 0.5+c_{z_1}-c_{w_1}-4.0 =c_{z_1}-c_{w_1}-3$ and we are done for (2).
  
For formula (3), we have ($\ps\leq 1$)
\begin{align*}
rO_n(i)-A(i)&\geq
3(c_{x_1}+c_{x_2}+c_{y_1}+c_{z_1}+\ps\theta)-(c_{x_1}+c_{x_2}+c_{y_1}+c_{w_1}+4.0)\\
&=2((c_{x_1}+c_{x_2}+c_{y_1}+c_{z_1})+3\ps\theta+c_{z_1}-c_{w_1}-4.0\\
&\geq 2\ps (c_{x_1}+c_{x_2}+c_{y_1}+c_{z_1}+\theta)+c_{z_1}-c_{w_1}-4.0.\\
\end{align*}
We again have the following time relations:\\
\hspace*{15mm} (i) $t_1<t_2$ ($M_S$ never turns on before $M_P$),\\
\hspace*{15mm} (ii) $t_3<t_1$ (see below),\\
\hspace*{15mm} (iii) $t_4-t_3\geq t_4-t_1$ (by (ii)) $\geq t_4-t_2$
(by (i)) $\geq 2/\ps$.\\
For (ii) recall that $M_P$ cannot execute $x_1$, $x_2$ and $y_1$ from time $t_1$. 
Since OPT does execute those jobs by a single machine, it should have started their execution before $t_1$, meaning $t_3<t_1$.
Since $c_{x_1}+c_{x_2}+c_{y_1}+c_{z_1}+\theta=t_4-t_3$, 
we finally have $$rO_n(i)-A(i)\geq 2\ps (2/\ps)+c_{z_1}-c_{w_1}-4.0=c_{z_1}-c_{w_1},$$
and (3) is proved. 
The generalization to arbitrary number of jobs is the same as before and may be omitted.
  
Finally we consider the case that $E_2$ (or even more)\ccc exists, which can be simply done by considering that 
a new virtual phase, just an interval starting from $t_4$ and ending at $t_e$ where we do the similar calculation as above.
Note that \oural\ needs only
one turn-on cost and so the energy consumption of the new virtual phase is the job execution costs $+3$\ccc
instead of $+4$ above. Therefore we do not need to lower bound the
cost for executing type-AO jobs (as we did for
$c_{x_1}+c_{x_2}+c_{y_1}$ above). Also, since the new virtual phase obviously includes
the whole idle time of $M_S$, the proof for formula
(3) is also straightforward.\ccc Details may be omitted.
  
Thus Lemma~3 is proved. 
\end{proof}

And the proof of Theorem~\ref{ubtheorem} is also concluded.
\end{proof}

\longdelete{
\section{Analysis of the Algorithm}
\label{analysis}
For analysis, we need to show that (i) \oural\ is correct, 
namely it can execute any sequence of jobs arriving under the single machine schedulability condition 
and (ii) its CR is at most 3.
We mostly focus on (ii).
(i) is not hard and the following observation should be enough.

The state of the system is described as A/B, where A is ON (OFF, resp.) if $M_P$ is ON (OFF, resp.). 
Similarly for B for $M_S$. 
Note that if B is ON, $M_S$ is busy, 
but if A is ON, $M_P$ is either busy or idle. 
The initial state is OFF/OFF.
When the due time of some job(s) comes, the state becomes ON/OFF. 
If no urgent jobs comes, the system returns to OFF/OFF after having processed a certain number of jobs on $M_P$ and its idle state has expired. 
Otherwise, if an urgent job comes, the system becomes ON/ON, and $M_P$ and $M_S$ are swapped. 
If no more jobs come, the system becomes ON/OFF or OFF/ON depending on which machine is finished first.
Otherwise, suppose a job comes when the system is ON/ON.
Then \oural\ first checks if $M_S$ is available and use it if possible. Otherwise \oural\ uses $M_P$ without any delay except one due to the EDF scheduling. 
When the system becomes ON/ON, $Q$ is empty and as we have just mentioned, no job goes to $Q$ when the system keeps ON/ON. 
Also note that if the whole input sequence satisfies Condition~\ref{input} of Lemma~\ref{assmpution_chen}, 
its arbitrary suffix obviously does, too.
Thus the availability of $M_P$ when the system is ON/ON is guaranteed.
Similarly for a new job when the system is OFF/ON.
$M_P$ is obviously available.

To prove the CR, we first define a {\em phase}. 
Suppose the system changes from OFF/OFF to ON/OFF at time $t_1$,
and returns to OFF/OFF for the first time at $t_e$.
Also let $a_1, a_2, \ldots, a_k$ be the arrival time of the jobs executed during $t_1$ to $t_e$ and let $t_0$ be $\min\{a_1, \ldots, a_k\}$. 
Then the time slot from $t_0$ to $t_e$ is called a phase.
Note that an entire execution of \oural\ consists of some phases.
Let $P(i)$ be the $i$'th phase.

\begin{figure}[htb]
\centering
\includegraphics[scale=0.6]{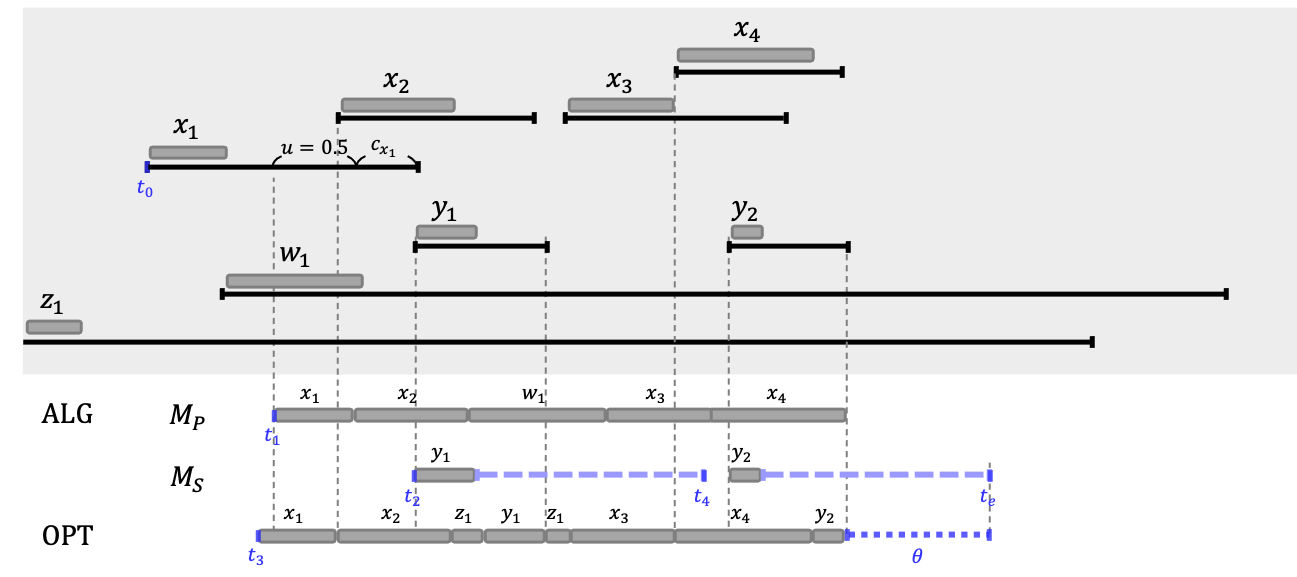}
\caption{A single phase. Now a period of idle time is $2/\ps$. $M_P$ and $M_S$ are swapped at~$a_{y_1}$.}
\label{fig:up5}
\end{figure}

Fig.~\ref{fig:up5} illustrates how a single phase looks like.
$x_1$ and $x_2$ are executed on $M_P$, 
then an urgent $y_1$ comes and executed on the new $M_P$ 
(so $w_1$, $x_3$ and $x_4$ are executed on the new $M_S$).
$w_1$ is a job moved from $Q$ to $Q_{M_P}$ when $x_1$ starts
(more precisely, $w_1$ is scheduled after $x_1$, 
but $x_2$ is inserted later due to the EDF principle).
$x_3$ and $x_4$ are added to $Q_{M_S}$. 
$y_2$ is a job that cannot be inserted to $Q_{M_S}$ because of the
existence of $x_4$; it goes to $M_P$.
Thus $M_P$ can be OFF for some period of time during a single phase 
(but $M_S$ should be ON during that period by definition). 
The earliest arrival time of the jobs \oural\ executes in this
phase is that of $x_1$, which is the start time, $t_0$, of this
phase.
Note that $t_e$ at which the phase ends is the moment when the idle time after $y_2$ expires.
An example of the OPT's execution sequence is given at the bottom of the figure. 
It must execute $x_1$, $x_2$, $x_3$, $x_4$, $y_1$ and $y_2$ since each of them has its arrival time and deadline within the phase.
$w_1$ is not executed by OPT in this example.
By contrast, $z_1$ which is performed by OPT is not executed by \oural\ (can be executed in the previous phase),
so its arrival time is not counted to determine the beginning of the phase.
(Note that OPT and \oural\ as well, may execute each job separately, so
we have two $z_1$'s in the figure.)
Fig.~\ref{fig:up5} is also used for the later analysis of the competitive ratio. 
A phase is called a single-machine phase if only one machine ($M_P$) is ON in that phase and a dual-machine phase otherwise.

Our proof of the CR uses a math induction.
Since \oural\ is deterministic, the set of jobs that are executed by \oural\ in $P(i)$ is uniquely determined once the input is given, 
which we denote by $J_A(i)$.  
For the jobs executed by OPT in $P(i)$, the situation is less clear; 
jobs whose arrival time and deadline are both within $P(i)$ must be executed, 
but jobs such that only one of them is within $P(i)$ may or may not be executed even partially. 
Furthermore, OPT may execute some jobs outside phases, 
namely while the system is OFF/OFF. 
Fix an arbitrary execution sequence $S(i)$ of OPT in this phase.
Then we can determine $\alpha(i)$ ($\lambda(i)$, and $\delta(i)$, resp.) $=$ the total execution time of the jobs executed by both \oural\ and OPT (only by \oural\, and only by OPT possibly partially, resp.) 
Now define three parameters for our math induction, 
$A(i)$, $O_f(i)$ and $O_n(i)$. 
$A(i)$ is the cost of \oural\ in $P(i)$ that includes $\alpha(i)$, $\lambda(i)$, turn-on costs and idle costs. 
$O_f(i)$ is a lower bound for the cost of OPT, 
namely it includes $\alpha(i)$, $\delta(i)$ and idle costs if any. 
Note that it does not include the turn-on cost when $S(i)$ starts since OPT may not need it depending on its state at the end of the previous phase, 
but does include one(s) if $S(i)$ includes turn-off and turn-on in the middle of the phase. 
$O_n(i)$ is similar but we impose the condition that OPT is ON at the end of $P(i)$. 
For instance, consider $P(i)$ whose execution sequence looks like Fig.~\ref{fig:up2-4} (b).
Then $A(i)=4.5$, $O_f(i)=0.5$ and $O_n(i)=2.5$,
where 0.5 is the execution of jobs and $2=(2/\ps)\ps$ is the idle cost to keep it ON until the end of the phase.

\begin{theorem}
\oural\ is correct and its CR is at most 3 for $u=0.5$.\ccc
\label{ubtheorem}
\end{theorem}

\begin{proof}
We omit the first part (see the previous observation).
For the CR, we fix an arbitrary execution sequence $S$ ($S(i)$ is
its subsequence associated with $P(i)$)
for the entire execution sequence of OPT and prove two lemmas.

\begin{lemma}\label{lemma3}
The following (2) and (3) hold for each phase for $r=3$.
\begin{align}
&rO_f(i)-A(i)\geq\delta(i)-\lambda(i)-r.\\
&rO_n(i)-A(i)\geq\delta(i)-\lambda(i).
\end{align}
\end{lemma}

The proof will be given later. 
Let $O(i)$ be the (real) cost of OPT under the sequence $S(i)$ in $P(i)$. 
Also let $m$ be an integer less than or equal to the number of phases.

\begin{lemma}
For $r=3$, we have
\begin{numcases}{\sum_{i=1}^{m}\left(rO(i)-A(i)\right)\geq}
\sum_{i=1}^{m}\left(\delta(i)-\lambda(i)\right)
\text{\hspace{3mm} if OPT is OFF at the end of $P(m)$.}\\
\sum_{i=1}^{m}\left(\delta(i)-\lambda(i)\right)+r
\text{\hspace{3mm} if OPT is ON at the end of $P(m)$.}
\end{numcases}
\end{lemma}
\begin{proof}
Suppose OPT is OFF at the end of $P(i)$. 
Then $O(i)$ should be at least $O_f(i)$ since the latter is a lower bound and similarly for $O_n(i)$ if OPT is ON at the end of the phase.
For $m=1$, note that OPT must spend the turn-on cost of 1 that is not included in either $O_f(i)$ or $O_n(i)$.
Therefore if OPT is OFF at the end of the phase, we have
\begin{align*}
&rO(1)-A(1)\geq rO_f(1)+r-A(1) \geq\delta(1)-\lambda(1)
\end{align*}
by Lemma 3.
Otherwise, if OPT is ON at the end of the phase, we have
\begin{align*}
&rO(1)-A(1)\geq rO_n(1)+r-A(1) \geq\delta(1)-\lambda(1)+r
\end{align*}
similarly. 
Now suppose the lemma is true for $m'=m-1$. 
Then to prove that the lemma also holds for $m'=m$, 
we consider four cases and define $C \rightarrow D$ as the states of OPT at the end of $P(m-1)$ and $P(m)$ respectively:
(i) OFF $\rightarrow$ OFF, (ii) OFF $\rightarrow$ ON, (iii) ON $\rightarrow$ OFF and (iv) ON $\rightarrow$ ON.
For case (i), OPT must pay the turn-on cost in $P(m)$, so
\begin{align*}
\sum_{i=1}^{m}\left(rO(i)-A(i)\right)&=\sum_{i=1}^{m-1}\left(rO(i)-A(i)\right)
+rO(m)-A(m)\\
&\geq\sum_{i=1}^{m-1}\left(\delta(i)-\lambda(i)\right) +0\cdot r
+rO_f(m)-A(m)+1\cdot r\\
&\geq\sum_{i=1}^{m-1}\left(\delta(i)-\lambda(i)\right)+\delta(m)-\lambda(m)\\
&=\sum_{i=1}^{m}\left(\delta(i)-\lambda(i)\right)
\end{align*}
by the induction hypothesis and Lemma 3.
Here 0 and 1 before $r$ are for handling the four cases. 
Since the current case is OFF $\rightarrow$ OFF, the first 0 means formula (4) does not have $r$ on the right-hand side and the second 1 means that OPT must turn on in $P(m)$.
The other cases are similar (just 0 and 1 before $r$ change) and may be omitted.
\end{proof}

If there are $m$ phases in total, it must be that
$\sum_{i=1}^{m}\left(\delta(i)-\lambda(i)\right)=0$ 
and thus the theorem is proved. 
What remains is to prove Lemma 3.

\begin{proof}
(Proof of Lemma 3)
Suppose $x_1=(a_1,d_1,c_1)$ is the first job executed in some phase $P$ starting from $t_0$ and ending at $t_k$. 
Then since $x_1$ is executed in $P$,
$t_0\leq a_1$ (there may be another job executed in $P$ and having an earlier arrival time).
Also, since the ending time of $x_1$'s execution is to be $d_1-0.5$
due to the delay
(or not delayed at all if this amount of delay is impossible)
and we have a mandatory idle time, $2/\ps\geq2$, after its (or a later job's) execution, 
it must be that $d_1\leq t_k$. 
This means the period ($a_1,d_1$) is included in $P$, meaning OPT also executes $x_1$ in $P$. 
Thus, in each phase, both \oural\ and OPT execute at least one job. 
Also it turns out, by definition, that \oural\ never executes a job outside phases.

Note, however, that OPT may execute some job, say $x$, outside phases
(this happens, e.g., if $u_1$ in Fig.~\ref{fig:up6} is executed by OPT after this phase and before the next phase). 
If that happens, we consider that $x$ is executed in a ``special'' phase, $P(i')$,
by extending the definition of a phase. 
Note that this phase has $A(i)=\lambda(i)=0$ and both $O_f(i)$ and $O_n(i)$ are at least $\delta(i)$, 
so (2) and (3) obviously hold. 
A special phase may continue to/from a neighboring (normal) phase.

\begin{figure}[htb]
\centering
\includegraphics[scale=0.6]{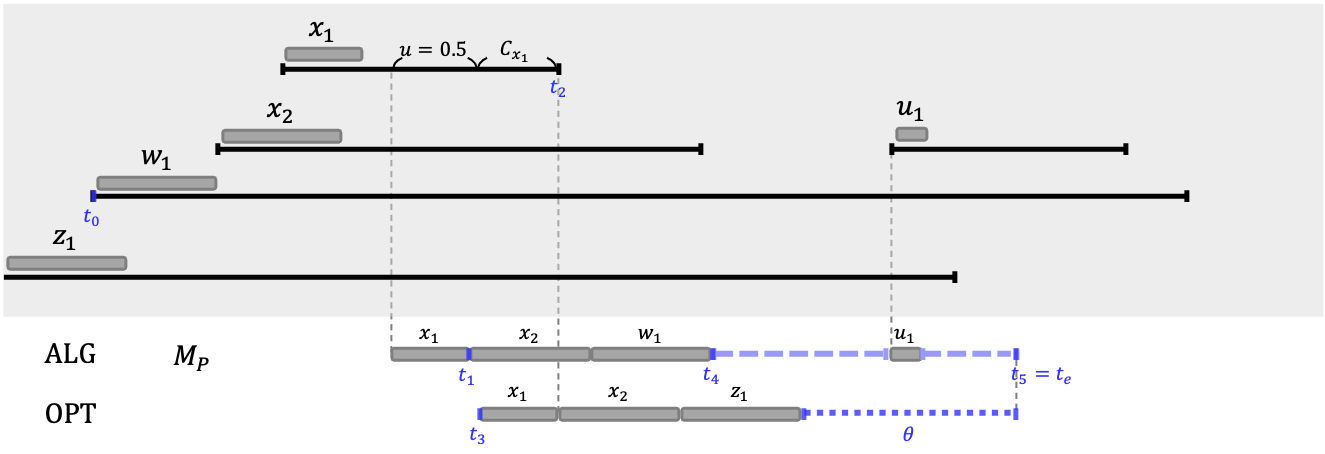}
\caption{A single-machine phase. Recall a period of idle time of \oural\ is $2/\ps$.}
\label{fig:up6}
\end{figure}

We first prove the lemma for a single-machine phase $P(i)$.
See Fig.~\ref{fig:up6} which illustrates the execution sequences of \oural\ and OPT as with Fig.~\ref{fig:up5}.
In this figure,
the phase starts at the arrival time of $w_1$ and ends when ALG's
machine ($M_P$) turns off.
$x_1$ is the first job executed in this phase,
and $x_2$ is a job whose arrival time and deadline are both within this phase 
(so must be executed by both \oural\ and OPT in this phase). 
$z_1$ has only its deadline, and
$u_1$ and $w_1$ have only their arrival times within the phase.
Thus there are several different types of jobs, 
but what is important is whether or not each job is executed in this phase.
In the example of this figure,
$x_1$ and $x_2$, called {\em type-AO} jobs, are executed by both \oural\ and OPT, 
$w_1$ and $u_1$, {\em type-A}, by only \oural, 
and $z_1$, {\em type-O}, only by OPT. 
Each type can include an arbitrary number of jobs, 
but as will be seen in a moment, 
those numbers are not important. 
So we will progress our analysis by using these five jobs, 
$x_1,x_2,w_1,u_1,z_1$, 
for a while and will mention the generalization after this analysis.
It should also be noted that a single job may be executed by OPT in two or more phases. 
If this happens, we divide that job into two or more parts and allow each part can have a different job type, if necessary.

Set the following moments (see the figure) of time: $t_1$ and $t_2$: 
the ending time of execution and the deadline of $x_1$, respectively.
$t_3$: the time when OPT becomes ON after time $t_0$ to execute $x_1$ or maybe another job.
$t_4$ and $t_5$: the start and ending times of the idle state of \oural, respectively. 
Thus one can observe that neither the number of jobs in each type nor their execution sequence is important to define these times,
except $x_1$ that is first executed.
For the five jobs, it turns out that $A(i)=c_{x_1}+c_{x_2}+c_{w_1}+c_{u_1}+1+2$ 
(1 is the turn-on cost and 2 is the idle cost; 
recall once \oural\ enters an idle state,
it continues until its total time becomes $2/\ps$,
i.e., until its total idle cost becomes $(2/\ps)\ps=2$),
and $O_f(i)\geq(c_{x_1}+c_{x_2}+c_{z_1})$.
OPT's execution of some job may exceed the border of the phase. 
If that happens, as mentioned before, 
we can partition that job into two parts, 
the first one ends at the end of the phase and the second one is to be the remaining part, 
which is executed in the following special phase.

Now we have
\begin{align*}
rO_f(i)-A(i)&\geq 3(c_{x_1}+c_{x_2}+c_{z_1})-(c_{x_1}+c_{x_2}+c_{w_1}+c_{u_1}+3)\\
&=2(c_{x_1}+c_{x_2}+c_{z_1})+c_{z_1}-(c_{w_1}+c_{u_1})-3\\
&\geq c_{z_1}-(c_{w_1}+c_{u_1})-3.
\end{align*}

Thus (2) holds for any nonnegative values of $c_{x_1}$ through $c_{u_1}$.
Similarly, letting $\theta$ be the idle time of OPT (if any) that keeps the OPT's machine ON until $t_5$ 
(OPT can turn off once and turn on at or before $t_5$, but OPT would then need an extra turn-on cost and our proof becomes easier), 
we have $O_n(i)\geq c_{x_1}+c_{x_2}+c_{z_1}+\ps \theta$ and hence
\begin{align*}
rO_n(i)-A(i)&\geq 3(c_{x_1}+c_{x_2}+c_{z_1}+\ps \theta)-(c_{x_1}+c_{x_2}+c_{w_1}+c_{u_1}+3)\\
&\geq 2(c_{x_1}+c_{x_2}+c_{z_1}+\ps \theta)+c_{z_1}-(c_{w_1}+c_{u_1})-3.
\end{align*}

Here we have the claim that $c_{x_1}+c_{x_2}+c_{z_1}+\ps \theta \geq
1.5$. Then we have 
$$rO_n(i)-A(i)\geq 2\times 1.5
+c_{z_1}-(c_{w_1}+c_{u_1})-3=c_{z_1}-(c_{w_1}+c_{u_1})$$
and formula (3) is also true. See Appendix for the proof of the claim,
for the generalization on the job set and for the proof for a
dual-machine phase. 
\end{proof} \end{proof}
}

\section{Lower Bound}
In this section we give our second result, a lower bound of 2.1, 
which improves 2.06 obtained in~\cite{chen2015online}.
Throughout this section we set $\ps=1$
(which seems the worst case for online algorithms).

\begin{theorem}
The CR of any online algorithm for the online DPM job scheduling problem is at least 2.1. 
\end{theorem}
\begin{proof}
Our strategy is quite simple and standard. 
The adversary, Adv, gives requests, one by one, 
so that each request 
blames
the last action of the algorithm. 
Fix an arbitrary algorithm ALG and a target CR lower bound, $\alpha$, 
we want to prove.
Before the formal proof, we briefly look at the
basic strategy of Adv. 
Recall that $\ps=1$ in this proof.

The first request by Adv is $j_1=(0,d_1,c_1)$,
where its execution time $c_1$ is tiny and $d_1$ should not be too
small.
ALG must execute $j_1$ at some time before $d_1$, say at $d_1-x_1$ on
one of the two machines, say $M_1$ (Fig.~\ref{fig:B-1}).
Here we have two cases.

\begin{figure}[htb]
\centering
\includegraphics[scale=0.55]{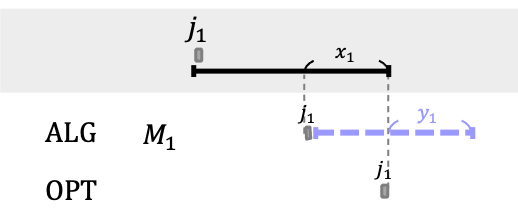}
\caption{Illustration of the incoming job $j_1$}
\label{fig:B-1}
\end{figure}

Suppose that $x_1$ is relatively large. Then Adv sees how long ALG
stays in the idle state after the execution of $j_1$. 
We can assume without loss of generality
that this idle state continues at least until 
$d_1$ and may be more for additional $y_1$ as shown in
Fig.~\ref{fig:B-1}.
(If the idle state ends before $d_1$, then
Adv immediately gives another tiny request with the same deadline
$d_1$. This situation is similar to that ALG postpones its execution
of $j_1$ up to this time. Since ALG has already spent a
turn-on cost of 1, it is not hard to show that Adv's job becomes
easier. We omit details.)
From the OPT side, it suffices to execute a tiny $c_1$ at $d_1$.
Here ALG cannot have a long $y_1$ since the CR at this moment is $\frac{1+c_1+x_1+y_1}{1+c_1}$ 
(both ALG and OPT need a turn-on cost of 1, since this is the beginning of the game),
which may exceed $\alpha$ and the game would end. 
So $y_1$ is relatively small, for which Adv gives a similar request right after the idle time expires. 
As shown in Fig.~\ref{fig:B-2}, 
OPT can manage these two requests by being ON from $d_1$ to $a_2$,
thus 
blaming
the small value of $y_1$.

\begin{figure}[htb]
\centering
\includegraphics[scale=0.55]{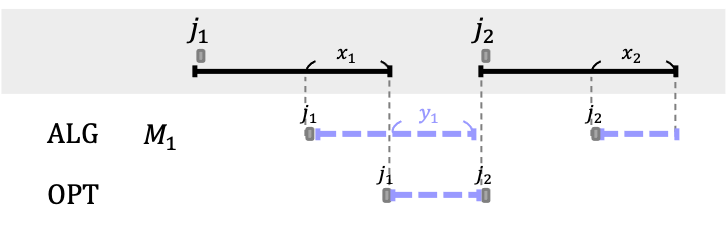}
\caption{Illustration of the incoming job $j_2$}
\label{fig:B-2}
\end{figure}

What if $x_1$ is relatively small? 
Then Adv gives a request $j_2$ as shown in Fig.~\ref{fig:A-2} immediately after ALG has started the execution of $j_1$.
Note that $j_2$ has the same deadline as $j_1$ (i.e. $d_1=d_2$) and its execution time is $x_1$ (i.e. $c_2=x_1$). 
Thus ALG cannot execute $j_2$ on $M_1$ and it turns on $M_2$ meaning ALG has to pay a new turn-on cost. 
OPT can manage this by being on from slightly before $d_1-x_1$ to $d_1$,
thus
blaming
the shortness of $x_1$.

\begin{figure}[htb]
\centering
\includegraphics[scale=0.55]{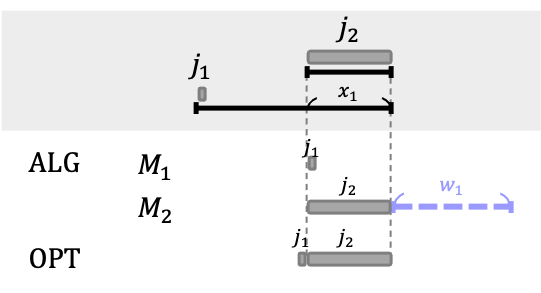}
\caption{Case A with $j_1$ and $j_2$}
\label{fig:A-2}
\end{figure}

Now we start our formal proof. As mentioned above, Adv has basically two different strategies depending on the first response of ALG, 
namely $x_1\leq\beta$ (Case A) and $x_1 > \beta$ (Case B), 
where $\beta$ is some constant to be optimized later. 
Let us look at Case A first.  
As shown above, 
it proceeds to the situation illustrated in Fig.~\ref{fig:A-2} and now ALG selects its new idle time $w_1$ on $M_2$.
Of course $M_1$ can also have some idle time. 
However, as seen in a moment, 
our Adv always gives a next request after both machines become OFF.  
Therefore without loss of generality,
we can assume only one machine which is busy until later than the other enters an idle state and the other turns off immediately when it finishes all the assigned requests. 
At the moment of Fig.~\ref{fig:A-2}, 
the CR is $\frac{2+c_1+x_1+w_1}{1+c_1+x_1}$ 
(recall $\ps=1$). Here we set $c_1=0$ for the exposition (and will do the same for a tiny execution time in the remaining part, too). 
This does not lose much sense since we can make $c_1$ arbitrarily small and it always appears as a sum with far greater values. 
Thus our current CR is $\frac{2+x_1+w_1}{1+x_1}$.
If this value is greater than $\alpha$, Adv has achieved its goal and the game ends. 
For the game to continue it must be $\frac{2+x_1+w_1}{1+x_1} \leq\alpha$. 
This implies
\begin{align*}
w_1\leq \alpha(1+x_1)-(2+x_1)
   \leq \alpha+(\alpha-1)x_1-2
   \leq \alpha+(\alpha-1)\beta-2 \; (=f_1)
\end{align*}
since $\alpha\geq1$ and $x_1\leq\beta$. 
Let $f_1$ be the value of the right-hand side.

\begin{figure}[htb]
\centering
\includegraphics[scale=0.55]{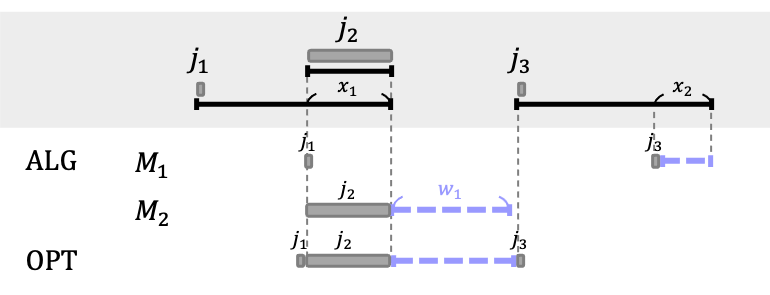}
\caption{Case A with $j_1, j_2$ and $j_3$}
\label{fig:A-3}
\end{figure}

Next, the adversary releases another tiny request $j_3$ after the idle
period of $M_2$. See Fig.~\ref{fig:A-3}.
At this moment, both $M_1$ and $M_2$ are OFF, so we can use $M_1$
without loss of generality for $j_3$.
Similarly as before, ALG executes $j_3$ at $d_3-x_2$ and has an idle time of $x_2$. 
OPT executes $j_3$ immediately when it 
comes by keeping
its idle state for $w_1$ assuming that $w_1\leq 1$, 
which can be verified after we eventually fix all parameter values
(indeed, $w_1\leq 0.63197$ for our final setting of the parameters).
Note that ALG needs three turn-ons and OPT one, 
so the current CR is $\frac{3+x_1+w_1+x_2}{1+x_1+w_1}$. 
For the game to be continued, it must be
\begin{align*}
x_2\leq \alpha(1+x_1+w_1)-(3+x_1+w_1)
\leq \alpha-3+(\alpha-1)\beta+(\alpha-1)w_1
\leq f_2+(\alpha-1)w_1
\end{align*}
by letting $\alpha-3+(\alpha-1)\beta = f_2$.

\begin{figure}[htb]
\centering
\includegraphics[scale=0.55]{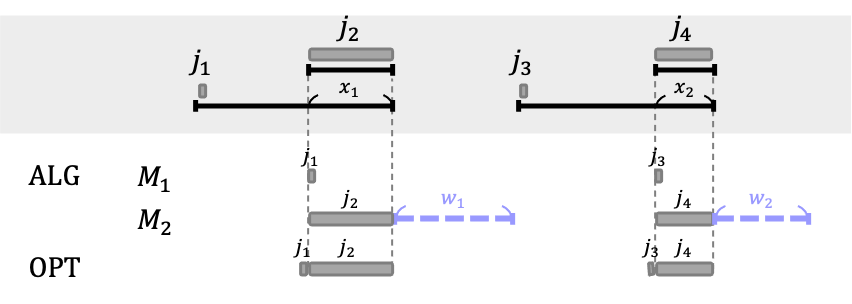}
\caption{Case A with $j_1$ to $j_4$}
\label{fig:A-4}
\end{figure}


The adversary then releases $j_4$ exactly as it did for $j_2$
as shown in Fig.~\ref{fig:A-4}.  
Here, it is better for OPT to turn
off at $d_2$ and turn on at $a_4$ since Adv selects a large $d_3-a_3$.
The CR is $\frac{4+x_1+w_1+x_2+w_2}{2+x_1+x_2}$.
For the game to continue, we must have
\begin{align*}
w_2 &\leq \alpha(2+x_1+x_2)-(4+x_1+w_1+x_2)\\
&= 2\alpha-4+(\alpha-1)x_1+(\alpha-1)x_2-w_1\\
&\leq 2\alpha-4+(\alpha-1)\beta+(\alpha-1)(f_2+(\alpha-1)w_1)-w_1\\
&= 2\alpha-4+(\alpha-1)\beta+(\alpha-1)f_2+(\alpha-1)^2 w_1-w_1\\
&= 2\alpha-4+(\alpha-1)(\beta+f_2)+(\alpha^2-2\alpha)w_1\\
&=f_3+(\alpha^2-2\alpha)w_1
\end{align*}
by letting $f_3=2\alpha-4+(\alpha-1)(\beta+f_2)$.

\begin{figure}[htb]
\centering
\includegraphics[scale=0.55]{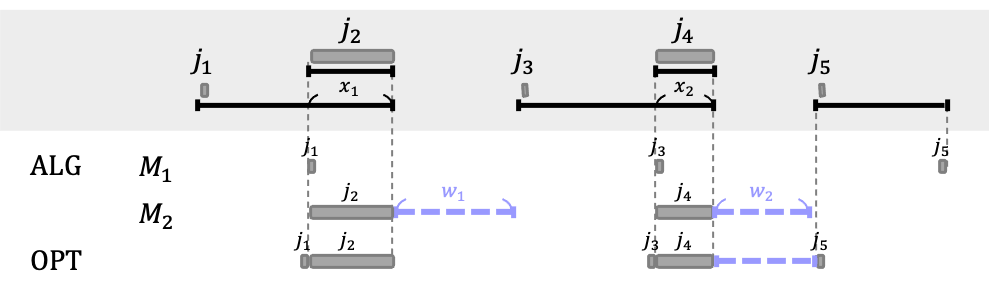}
\caption{Case A ending with $j_5$}
\label{fig:A-5}
\end{figure}

Finally,
the adversary releases a tiny request $j_5$ as shown in
Fig.~\ref{fig:A-5}. 
For ALG, we impose only a turn-on cost which ALG at least spends. 
OPT can manage this request by continuing its idle state as
before. Now the CR is at least
\begin{align*}
\frac{5+x_1+w_1+x_2+w_2}{2+x_1+x_2+w_2}&\geq
\frac{5+\beta+w_1+f_2+(\alpha-1)w_1+f_3+(\alpha^2-2\alpha)w_1}{2+\beta+f_2+(\alpha-1)w_1+f_3+(\alpha^2-2\alpha)w_1}\\
&=\frac{5+\beta+f_2+f_3+(\alpha^2-\alpha)w_1}{2+\beta+f_2+f_3+(\alpha^2-\alpha-1)w_1}\\
&\geq
\frac{5+\beta+f_2+f_3+(\alpha^2-\alpha)f_1}{2+\beta+f_2+f_3+(\alpha^2-\alpha-1)f_1}=C_A.
\end{align*}
The first inequality holds for the following reason:
$x_1$, $x_2$ and $w_2$ appear both in the numerator and the
denominator, the fraction becomes minimum when all $x_1$,
$x_2$ and $w_2$ are maximum. For the second inequality,
recall that our target CR, $\alpha$, is greater than 2. So
$\frac{(\alpha^2-\alpha)}{\alpha^2-\alpha-1}\leq2$, which means the
fraction becomes minimum when $w_1$ is maximum. 
Thus the CR is at least
$C_A$ for Case A.




We now look at Case B, namely $x_1>\beta$.
After the request $j_1$ (Fig.~\ref{fig:B-1}), 
our CR is $\frac{1+x_1+y_1}{1}$ 
(recall $\ps=1$ and\ccc that we ignore the tiny execution time). 
Thus for the game to continue, it must be
\begin{align*}
y_1\leq \alpha-1-x_1 \leq \alpha-1-\beta \; (=g_1).
\end{align*}

Then the adversary releases the second request, $j_2$, 
as shown in Fig.~\ref{fig:B-2}, 
which is very similar to $j_3$ of Case A.
The CR is $\frac{2+x_1+y_1+x_2}{1+y_1}$. 
Thus for the game to continue, it must be
\begin{align*}
x_2 &\leq \alpha(1+y_1)-(2+x_1+y_1)\\
&=(\alpha-1)y_1+\alpha-2-x_1\\
&\leq(\alpha-1)y_1+\alpha-2-\beta\leq (\alpha-1)y_1+g_2 \;
(g_2=\alpha-2-\beta).
\end{align*}\ccc

\begin{figure}[htb]
\centering
\includegraphics[scale=0.55]{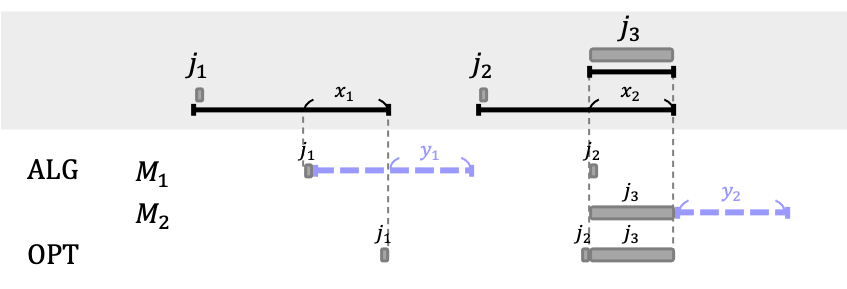}
\caption{Case B with $j_1, j_2$ and $j_3$}
\label{fig:B-3}
\end{figure}

The next request $j_3$ is similar to $j_4$ of Case A. 
See Fig.~\ref{fig:B-3}.
Our CR is $\frac{3+x_1+y_1+x_2+y_2}{2+x_2}$. 
To continue the game, we need to have
\begin{align*}
y_2 &\leq \alpha(2+x_2)-(3+x_1+y_1+x_2)\\
&= (\alpha-1)x_2+2\alpha-3-x_1-y_1\\
&\leq (\alpha-1)((\alpha-1)y_1+g_2)+2\alpha-3-\beta-y_1\\
&= (\alpha-1)^2 y_1+(\alpha-1)g_2+2\alpha-3-\beta-y_1\\
&= (\alpha^2-2\alpha)y_1+(\alpha-1)g_2++2\alpha-3-\beta\\
&= (\alpha^2-2\alpha)y_1+g_3
\end{align*}
by letting $g_3=(\alpha-1)g_2+2\alpha-3-\beta$.

\begin{figure}[htb]
\centering
\includegraphics[scale=0.55]{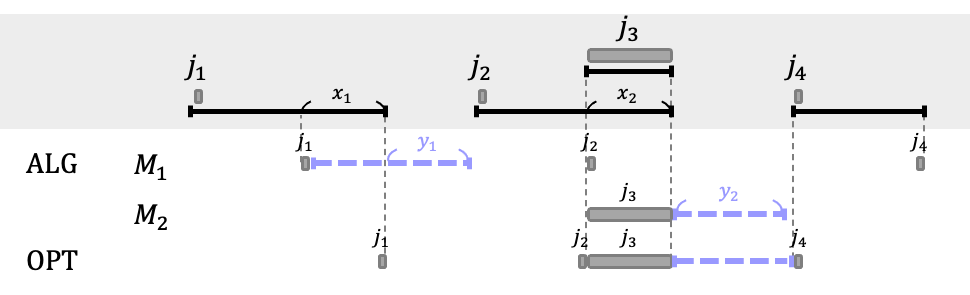}
\caption{Case B ending with $j_4$}
\label{fig:B-4}
\end{figure}

The last request is $j_4$. See Fig.~\ref{fig:B-4}. 
Our CR is
\begin{align*}
\frac{4+x_1+y_1+x_2+y_2}{2+x_2+y_2}&\geq
\frac{4+\beta+y_1+(\alpha-1)y_1+g_2+(\alpha^2-2\alpha)y_1+g_3}{2+(\alpha-1)y_1+g_2+(\alpha^2-2\alpha)y_1+g_3}\\
&=\frac{(\alpha^2-\alpha)y_1+4+\beta+g_2+g_3}{(\alpha^2-\alpha-1)y_1+2+g_2+g_3}.
\end{align*}
Implying inequalities is pretty much the same as before.\ccc
To make the first fraction minimum, 
$x_1$ should be minimum since it appears only in the numerator and $x_2$ and $y_2$ maximum since they appear in both. 
And, again, since
$\frac{(\alpha^2-\alpha)}{\alpha^2-\alpha-1}\leq2$, 
the fraction is minimum for maximum $y_1$.  
Thus the CR is at least
\begin{align*}
C_B= \frac{(\alpha^2-\alpha)g_1+4+\beta+g_2+g_3}{(\alpha^2-\alpha-1)g_1+2+g_2+g_3}.
\end{align*}

For $\beta=0.4745$ and $\alpha=2.1068$, 
our numerical calculation shows that $C_A=2.107447$ and $C_B=2.106989$, 
and the theorem is proved.
\end{proof}

\section{Concluding Remarks}
Obvious future work is to narrow the gap. For the lower bound, 
one can easily notice that increasing the number of stages (currently, it is three) might help. 
This is correct and in fact one more stage can give us a strictly better bound. 
Unfortunately the degree of improvement is already pretty small and getting even smaller in further stages. 
For the upper bound, it is obviously important to consider more flexible structures for delaying requests. 
Our present algorithm executes all the pending requests when one request comes to its due time.
It would be even more important to make our CR a function in $\ps$.  
Our open question is that a smaller $\ps$ very likely implies a better CR.\ccc

\newpage 
\nocite{*}
\bibliography{bibliography}


\longdelete{
\subsection{Time Interval}
In~\cite{irani2007algorithms, chen2015online}, they proved that a time interval in which the status of the optimal offline schedule is OFF (which they call a \emph{sleep-interval}) intersects at most two time intervals in which the status of their online algorithms
is ON (which they call \emph{awaken-intervals}).
See Fig.~\ref{fig:intersection} for a more intuitive description.
The first intersection appears when OPT becomes OFF earlier than ALG as ALG needs a fixed idle time after the assigned jobs are finished. 
On the contrary, OPT can be switched off immediately when the jobs are done. 
Then, the second intersection comes when ALG becomes ON earlier than OPT since ALG is designed to be turned on to execute some job(s) to avoid the scary adversary. 
Relatively, as OPT knows all input jobs in advance, there is a chance for OPT to perform the job(s) later than ALG. 
In this way, one OFF state of OPT is intersected at most two ON states of ALG.
In particular, the first and last OFF state for the entire OPT execution sequence will only intersect at most one ON state for ALG.

\begin{figure}[htb]
\centering
\includegraphics[scale=0.7]{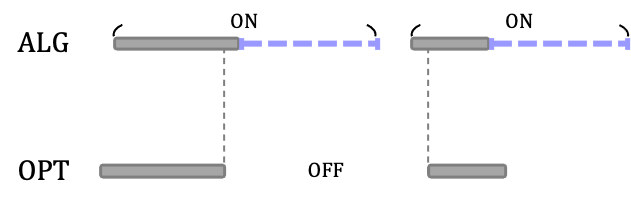}
\caption{An illustration of the intersection between the interval of the OFF status of OPT and the two intervals of the ON status of their ALG in \cite{irani2007algorithms,chen2015online}}
\label{fig:intersection}
\end{figure}

As mentioned, 
Chen et al.~\cite{chen2015online} derived the competitive ratio by analyzing the relative positions between the time intervals of OPT and ALG. 
Then they charged the energy consumed by their designed schedule to that consumed by the optimal offline schedule to obtain the claimed bound.
However, by the argument mentioned above, the sleep interval of OPT is counted at most twice when considering a single awaken interval of ALG, which results in a 4-competitive schedule.
This is a genuinely classical technique for the class of problems.
Motivated by the design of their time intervals~\cite{irani2007algorithms, chen2015online}, 
we cut the time line of our algorithm into phases, where all the input jobs and situations have been considered and no doubt of repeated calculation (see Section~\ref{UB} for more details).
}
  
\longdelete{
\subsection{Remaining Part of the Proof of Lemma 3}

\begin{clm}
$$c_{x_1}+c_{x_2}+c_{z_1}+\ps \theta \geq 1.5$$
\label{claim}
\end{clm}
\begin{proof}
We have a sequence of time relations:\\
\hspace*{15mm} (i) $t_2-t_1\leq 0.5$ due to the margin,\\
\hspace*{15mm} (ii) $t_3\leq t_2$ ($x_1$ must be executed before its deadline),\\
\hspace*{15mm} (iii) $t_1\leq t_4$ (obviously),\\
\hspace*{15mm} (iv) $t_3-t_4\leq t_3-t_1$ (by (iii)) $\leq
t_2-t_1$ (by (ii)) $\leq 0.5$ (by (i)),\\
\hspace*{15mm} (v) $t_5-t_4\geq 2/\ps$ (the total idle time),\\
\hspace*{15mm} (vi) $t_4-t_5\leq -2/\ps$ (inversion of (v)), and\\
\hspace*{15mm} (vii) $c_{x_1}+c_{x_2}+c_{z_1}+\theta=t_5-t_3=-(t_3-t_5)\geq 2/\ps -0.5$ (by (iv)+(vi)).\\
We may need a bit more explanation for (i). Recall that if $x_1$ is
delayed and nothing happens, $t_2-t_1$ is exactly 0.5. However, some
(small) job that comes after $x_1$'s execution has started
may interrupt $x_1$ and be inserted. It should also be noted that if
$d_{x_1}-(a_{x_1}+c_{x_1}) < 0.5$, then $x_1$ is not delayed in the
first place. Now we have from (vii) (note $\ps\leq 1$)\ccc
\hspace{20mm} $$c_{x_1}+c_{x_2}+c_{z_1}+\ps\theta \geq \ps(c_{x_1}+c_{x_2}+c_{z_1}+\theta)
\geq 2-0.5\ps \geq 1.5.$$
\end{proof}
  
Thus we have
\begin{align*}
rO_n(i)-A(i)&\geq 2\times 1.5 +c_{z_1}-(c_{w_1}+c_{u_1})-3=c_{z_1}-(c_{w_1}+c_{u_1}),
\end{align*}
meaning (3) is also true.
In general, we have more (or less) jobs in each job type. 
However,
one can see that the definitions of $t_1$ through $t_5$ are not affected by that (only the first job, $x_1$ is important). 
Also we can simply replace $c_{x_1}$ by the sum of execution times of type-AO jobs and similarly for type-A and type-O jobs. 
Thus the extension to the general case is straightforward and details may be omitted.

We next consider a dual-machine phase as shown in Fig.~\ref{fig:up5},
where\ccc there are two executions, $y_1$ and $y_2$ on $M_P$
while $M_S$ is busy. (Recall $M_P$ and $M_S$ are swapped when $M_S$
turns on. Now $M_S$ is busy and it may not be available for urgent jobs.)\ccc
Let the first execution be $E_1$ and the second one $E_2$.
We first consider the case that $E_2$ does not exist. Namely, there
is no $x_3$, $x_4$ or $y_2$ and OPT has an idle time after having
executed $z_1$.
Thus the phase would be ending at the end of the idle state following
the execution of $y_1$\ccc and we prove the lemma for this case first. 
As before, we use specific examples for jobs to be executed in this phase,
$x_1$, $x_2$ and $y_1$ for type-AO, $w_1$ for type-A and $z_1$ for
type-O. Since the $M_S$'s turn-on is provoked, there must be a set $S$
of jobs such that their execution on $M_P$ is impossible
during the period from some $t'_1$ to some $t'_2$.\ccc
We assume that $S=\{x_1,x_2,y_1\}$, where $t'_1=t_1$ and $t'_2$ is the
deadline of $y_1$.\ccc
What we do for the generalization is the same as before, 
namely we replace $\{x_1,x_2,y_1\}$ by the real jobs in $S$, 
maybe more jobs for type-AO, replace $w_1$ by real type-A jobs, and $z_1$ by real type-O jobs.
  
Now we start with definitions of time moments (see the figure) as before. $t_1$, $t_2$ and $t_3$: the time when $M_P$ and $M_S$ become on and the time OPT becomes busy, respectively. 
$t_4$: the time when the idle state of $M_P$ is ended and this is
the end of the phase, too.\ccc
For those five jobs, we have
$A(i)=c_{x_1}+c_{x_2}+c_{y_1}+c_{w_1}+2+2$ (we now need to turn on
both $M_P$ and $M_S$), $O_f(i)\geq (c_{x_1}+c_{x_2}+c_{y_1}+c_{z_1})$,
and $O_n(i)\geq (c_{x_1}+c_{x_2}+c_{y_1}+c_{z_1}+\ps \theta)$, where $\theta$ is
the idle time (if any) to keep the OPT's machine on until $t_4$.
  
To prove formulas (2), we have (recall $\ps\leq 1$)
\begin{align*}
rO_f(i)-A(i)&\geq 3(c_{x_1}+c_{x_2}+c_{y_1}+c_{z_1})-(c_{x_1}+c_{x_2}+c_{y_1}+c_{w_1}+4.0)\\
&=2(c_{x_1}+c_{x_2}+c_{y_1}+c_{z_1})+c_{z_1}-c_{w_1}-4.0.\\
\end{align*}
Here we can claim that $c_{x_1}+c_{x_2}+c_{y_1}\geq 0.5$. 
The reason is that there is a margin of 0.5 between the end of the
execution of $x_1$ and its deadline (recall again if realizing this
margin is impossible, $M_S$ would not have been turned on).\ccc
So if $c_{x_1}+c_{x_2}+c_{y_1}< 0.5$, 
then it follows of course $c_{x_2}+c_{y_1}< 0.5$, 
which means that $x_2$ and $y_1$ could have been executed using this margin time 
(they can interrupt the execution of $x_1$) on $M_P$, resulting in a contradiction. 
Thus $rO_f(i)-A(i)\geq 2\times 0.5+c_{z_1}-c_{w_1}-4.0 =c_{z_1}-c_{w_1}-3$ and we are done for (2).
  
For formula (3), we have ($\ps\leq 1$)
\begin{align*}
rO_n(i)-A(i)&\geq
3(c_{x_1}+c_{x_2}+c_{y_1}+c_{z_1}+\ps\theta)-(c_{x_1}+c_{x_2}+c_{y_1}+c_{w_1}+4.0)\\
&=2((c_{x_1}+c_{x_2}+c_{y_1}+c_{z_1})+3\ps\theta+c_{z_1}-c_{w_1}-4.0\\
&\geq 2\ps (c_{x_1}+c_{x_2}+c_{y_1}+c_{z_1}+\theta)+c_{z_1}-c_{w_1}-4.0.\\
\end{align*}
We again have the following time relations:\\
\hspace*{15mm} (i) $t_1<t_2$ ($M_S$ never turns on before $M_P$),\\
\hspace*{15mm} (ii) $t_3<t_1$ (see below),\\
\hspace*{15mm} (iii) $t_4-t_3\geq t_4-t_1$ (by (ii)) $\geq t_4-t_2$
(by (i)) $\geq 2/\ps$.\\
For (ii) recall that $M_P$ cannot execute $x_1$, $x_2$ and $y_1$ from time $t_1$. 
Since OPT does execute those jobs by a single machine, it should have started their execution before $t_1$, meaning $t_3<t_1$.
Since $c_{x_1}+c_{x_2}+c_{y_1}+c_{z_1}+\theta=t_4-t_3$, 
we finally have $$rO_n(i)-A(i)\geq 2\ps (2/\ps)+c_{z_1}-c_{w_1}-4.0=c_{z_1}-c_{w_1},$$
and (3) is proved. 
The generalization to arbitrary number of jobs is the same as before and may be omitted.
  
Finally we consider the case that $E_2$ (or even more)\ccc exists, which can be simply done by considering that 
a new virtual phase, just an interval starting from $t_4$ and ending at $t_e$ where we do the similar calculation as above.
Note that \oural\ needs only
one turn-on cost and so the energy consumption of the new virtual phase is the job execution costs $+3$\ccc
instead of $+4$ above. Therefore we do not need to lower bound the
cost for executing type-AO jobs (as we did for
$c_{x_1}+c_{x_2}+c_{y_1}$ above). Also, since the new virtual phase obviously includes
the whole idle time of $M_S$, the proof for formula
(3) is also straightforward.\ccc Details may be omitted.
  
Thus Lemma~3 is proved.
  
And the proof of Theorem~\ref{ubtheorem} is also concluded.
}

\longdelete{
\subsection{Proof of Theorem 5}

Our strategy is quite simple and standard. 
The adversary, Adv, gives requests, one by one, 
so that each request 
blames
the last action of the algorithm. 
Fix an arbitrary algorithm ALG and a target CR lower bound, $\alpha$, 
we want to prove.
Before the formal proof, we briefly look at the
basic strategy of Adv. Throughout this proof, $\ps=1$, i.e., the idle
time is equal to the idle cost.\ccc

The first request by Adv is $j_1=(0,d_1,c_1)$,
where its execution time $c_1$ is tiny and $d_1$ should not be too
small.
ALG must execute $j_1$ at some time before $d_1$, say at $d_1-x_1$ on
one of the two machines, say $M_1$.
Here we have two cases. 

Suppose that $x_1$ is relatively large. Then Adv sees how long ALG
stays in the idle state after the execution of $j_1$. 
We can assume without loss of generality
that this idle state continues at least until 
$d_1$ and may be more for additional $y_1$ as shown in
Fig.~\ref{fig:LB} (a).
(If the idle state ends before $d_1$, then
Adv immediately gives another tiny request with the same deadline
$d_1$. This situation is similar to that ALG postpones its execution
of $j_1$ up to this time. Since ALG has already spent a
turn-on cost of 1, it is not hard to show that Adv's job becomes
easier. We omit details.)\ccc
From the OPT side, it suffices to execute a tiny $c_1$ at $d_1$.
Here ALG cannot have a long $y_1$ since the CR at this moment is $\frac{1+c_1+x_1+y_1}{1+c_1}$ 
(both ALG and OPT need a turn-on cost of 1, since this is the beginning of the game),
which may exceed $\alpha$ and the game would end. 
So $y_1$ is relatively small, for which Adv gives a similar request right after the idle time expires. 
As shown in Fig.~\ref{fig:B-2}, 
OPT can manage these two requests by being on from $d_1$ to $a_2$,
thus 
blaming
the small value of $y_1$.

\begin{figure}[htb]
\centering
\includegraphics[scale=0.55]{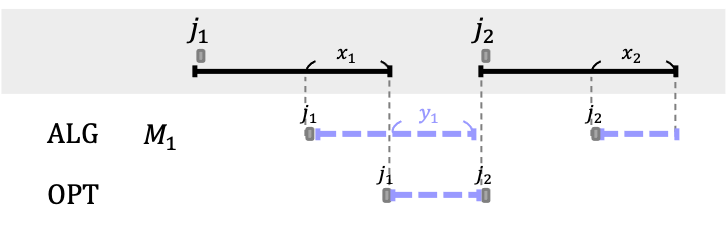}
\caption{Case B with $j_1$ and $j_2$}
\label{fig:B-2}
\end{figure}

What if $x_1$ is relatively small? 
Then Adv gives a request $j_2$ as shown in Fig.~\ref{fig:A-2} immediately after ALG has started the execution of $j_1$.
Note that $j_2$ has the same deadline as $j_1$ (i.e. $d_1=d_2$) and its execution time is $x_1$ (i.e. $c_2=x_1$). 
Thus ALG cannot execute $j_2$ on $M_1$ and it turns on $M_2$ meaning ALG has to pay a new turn-on cost. 
OPT can manage this by being on from slightly before $d_1-x_1$ to $d_1$,
thus
blaming
the shortness of $x_1$.\ccc

\begin{figure}[htb]
\centering
\includegraphics[scale=0.55]{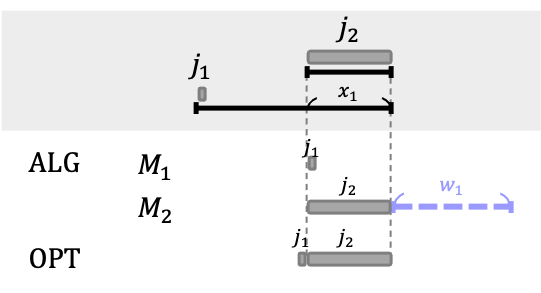}
\caption{Case A with $j_1$ and $j_2$}
\label{fig:A-2}
\end{figure}

Now we start our formal proof. As mentioned above,\ccc
Adv has basically two different strategies depending on the first response of ALG, 
namely $x_1\leq\beta$ (Case A) and $x_1 > \beta$ (Case B), 
where $\beta$ is some constant to be optimized later. 
Let us look at Case A first.  
As shown above, 
it proceeds to the situation illustrated in Fig.~\ref{fig:A-2} and now ALG selects its new idle time $w_1$ on $M_2$.
Of course $M_1$ can also have some idle time. 
However, as seen in a moment, 
our Adv always gives a next request after both machines become OFF.  
Therefore without loss of generality,
we can assume only one machine which is busy until later than the other enters an idle state and the other turns off immediately when it finishes all the assigned requests. 
At the moment of Fig.~\ref{fig:A-2}, 
the CR is $\frac{2+c_1+x_1+w_1}{1+c_1+x_1}$ 
(recall $\ps=1$). Here\ccc we set $c_1=0$ for the exposition (and will do the same for a tiny execution time in the remaining part, too). 
This does not lose much sense since we can make $c_1$ arbitrarily small and it always appears as a sum with far greater values. 
Thus our current CR is $\frac{2+x_1+w_1}{1+x_1}$.
If this value is greater than $\alpha$, Adv has achieved its goal and the game ends. 
For the game to continue it must be $\frac{2+x_1+w_1}{1+x_1} \leq\alpha$. 
This implies
\begin{align*}
w_1\leq \alpha(1+x_1)-(2+x_1)
   \leq \alpha+(\alpha-1)x_1-2
   \leq \alpha+(\alpha-1)\beta-2 \; (=f_1)
\end{align*}
since $\alpha\geq1$ and $x_1\leq\beta$. 
Let $f_1$ be the value of the right-hand side.

\begin{figure}[htb]
\centering
\includegraphics[scale=0.55]{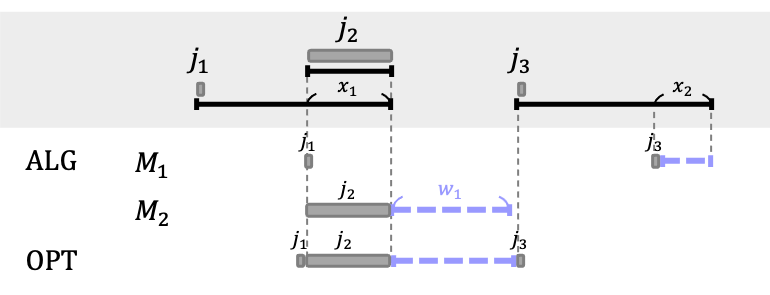}
\caption{Case A with $j_1, j_2$ and $j_3$}
\label{fig:A-3}
\end{figure}

Next, the adversary releases another tiny request $j_3$ after the idle
period of $M_2$. See Fig.~\ref{fig:A-3}.
At this moment, both $M_1$ and $M_2$ are OFF, so we can use $M_1$
without loss of generality for $j_3$.\ccc
Similarly as before, ALG executes $j_3$ at $d_3-x_2$ and has an idle time of $x_2$. 
OPT executes $j_3$ immediately when it 
comes by keeping
its idle state for $w_1$ assuming that $w_1\leq 1$, 
which can be verified after we eventually fix all parameter values
(indeed, $w_1\leq 0.63197$ for our final setting of the parameters).\ccc
Note that ALG needs three turn-ons and OPT one, 
so the current CR is $\frac{3+x_1+w_1+x_2}{1+x_1+w_1}$. 
For the game to be continued, it must be
\begin{align*}
x_2\leq \alpha(1+x_1+w_1)-(3+x_1+w_1)
\leq \alpha-3+(\alpha-1)\beta+(\alpha-1)w_1
\leq f_2+(\alpha-1)w_1
\end{align*}
by letting $\alpha-3+(\alpha-1)\beta = f_2$.

\begin{figure}[htb]
\centering
\includegraphics[scale=0.55]{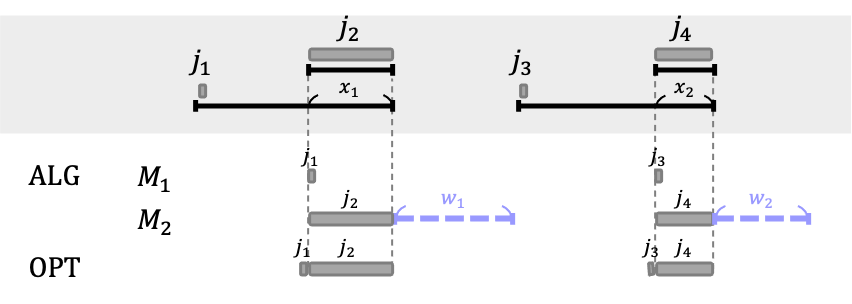}
\caption{Case A with $j_1$ to $j_4$}
\label{fig:A-4}
\end{figure}


The adversary then releases $j_4$ exactly as it did for $j_2$
as shown in Fig.~\ref{fig:A-4}.  
Here, it is better for OPT to turn
off at $d_2$ and turn on at $a_4$ since Adv selects a large $d_3-a_3$.
The CR is $\frac{4+x_1+w_1+x_2+w_2}{2+x_1+x_2}$.
For the game to continue, we must have
\begin{align*}
w_2 &\leq \alpha(2+x_1+x_2)-(4+x_1+w_1+x_2)\\
&= 2\alpha-4+(\alpha-1)x_1+(\alpha-1)x_2-w_1\\
&\leq 2\alpha-4+(\alpha-1)\beta+(\alpha-1)(f_2+(\alpha-1)w_1)-w_1\\
&= 2\alpha-4+(\alpha-1)\beta+(\alpha-1)f_2+(\alpha-1)^2 w_1-w_1\\
&= 2\alpha-4+(\alpha-1)(\beta+f_2)+(\alpha^2-2\alpha)w_1\\
&=f_3+(\alpha^2-2\alpha)w_1
\end{align*}
by letting $f_3=2\alpha-4+(\alpha-1)(\beta+f_2)$.

\begin{figure}[htb]
\centering
\includegraphics[scale=0.55]{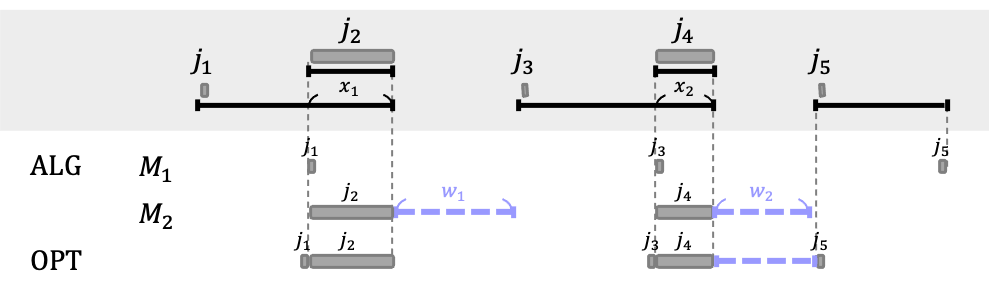}
\caption{Case A ending with $j_5$}
\label{fig:A-5}
\end{figure}

Finally,
the adversary releases a tiny request $j_5$ as shown in
Fig.~\ref{fig:A-5}. 
For ALG, we impose only\ccc a turn-on cost which ALG at least spends. 
OPT can manage this request by continuing its idle state as
before. Now the CR is at least
\begin{align*}
\frac{5+x_1+w_1+x_2+w_2}{2+x_1+x_2+w_2}&\geq
\frac{5+\beta+w_1+f_2+(\alpha-1)w_1+f_3+(\alpha^2-2\alpha)w_1}{2+\beta+f_2+(\alpha-1)w_1+f_3+(\alpha^2-2\alpha)w_1}\\
&=\frac{5+\beta+f_2+f_3+(\alpha^2-\alpha)w_1}{2+\beta+f_2+f_3+(\alpha^2-\alpha-1)w_1}.\\
&\geq
\frac{5+\beta+f_2+f_3+(\alpha^2-\alpha)f_1}{2+\beta+f_2+f_3+(\alpha^2-\alpha-1)f_1}=C_A
\end{align*}
The first inequality holds for the following reason:
$x_1$, $x_2$ and $w_2$ appear both in the numerator and the
denominator, the fraction becomes minimum when all $x_1$,
$x_2$ and $w_2$ are maximum. For the second inequality,
recall that our target CR, $\alpha$, is greater than 2. So
$\frac{(\alpha^2-\alpha)}{\alpha^2-\alpha-1}\leq2$, which means the
fraction becomes minimum when $w_1$ is maximum. Thus the CR is at least
$C_A$ for Case A.

We now look at Case B, namely $x_1>\beta$.
After the request $j_1$ (Fig.~\ref{fig:LB} (a)), 
our CR is $\frac{1+x_1+y_1}{1}$ 
(recall $\ps=1$ and\ccc that we ignore the tiny execution time). 
Thus for the game to continue, it must be
\begin{align*}
y_1\leq \alpha-1-x_1 \leq \alpha-1-\beta \; (=g_1).
\end{align*}

Then the adversary releases the second request, $j_2$, 
as shown in Fig.~\ref{fig:B-2}, 
which is very similar to $j_3$ of Case A.
The CR is $\frac{2+x_1+y_1+x_2}{1+y_1}$. 
Thus for the game to continue, it must be
\begin{align*}
x_2 &\leq \alpha(1+y_1)-(2+x_1+y_1)\\
&=(\alpha-1)y_1+\alpha-2-x_1\\
&\leq(\alpha-1)y_1+\alpha-2-\beta\leq (\alpha-1)y_1+g_2 \;
(g_2=\alpha-2-\beta).
\end{align*}

\begin{figure}[htb]
\centering
\includegraphics[scale=0.55]{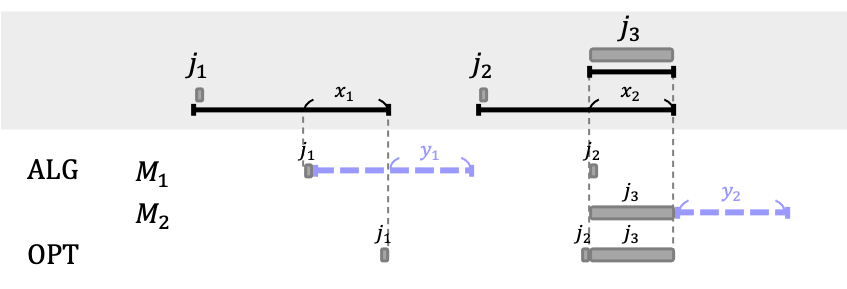}
\caption{Case B with $j_1, j_2$ and $j_3$}
\label{fig:B-3}
\end{figure}

The next request $j_3$ is similar to $j_4$ of Case A. 
See Fig.~\ref{fig:B-3}.
Our CR is $\frac{3+x_1+y_1+x_2+y_2}{2+x_2}$. 
To continue the game, we need to have
\begin{align*}
y_2 &\leq \alpha(2+x_2)-(3+x_1+y_1+x_2)\\
&= (\alpha-1)x_2+2\alpha-3-x_1-y_1\\
&\leq (\alpha-1)((\alpha-1)y_1+g_2)+2\alpha-3-\beta-y_1\\
&= (\alpha-1)^2 y_1+(\alpha-1)g_2+2\alpha-3-\beta-y_1\\
&= (\alpha^2-2\alpha)y_1+(\alpha-1)g_2++2\alpha-3-\beta\\
&= (\alpha^2-2\alpha)y_1+g_3
\end{align*}
by letting $g_3=(\alpha-1)g_2+2\alpha-3-\beta$.

The last request is $j_4$. See Fig.~\ref{fig:B-4}. 
\begin{figure}[htb]
\centering
\includegraphics[scale=0.55]{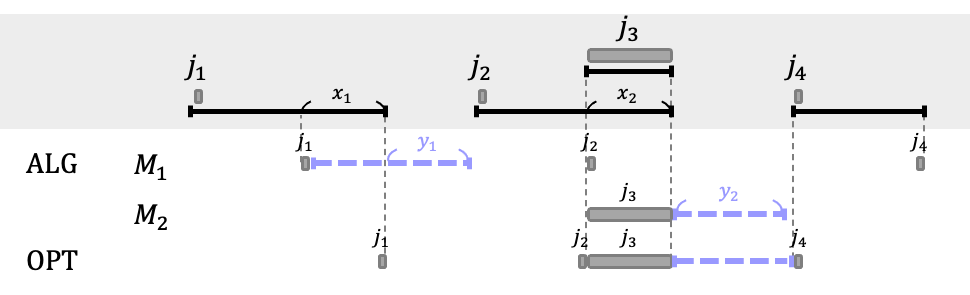}
\caption{Case B ending with $j_4$}
\label{fig:B-4}
\end{figure}
Our CR is
\begin{align*}
\frac{4+x_1+y_1+x_2+y_2}{2+x_2+y_2}&\geq
\frac{4+\beta+y_1+(\alpha-1)y_1+g_2+(\alpha^2-2\alpha)y_1+g_3}{2+(\alpha-1)y_1+g_2+(\alpha^2-2\alpha)y_1+g_3}\\
&=\frac{(\alpha^2-\alpha)y_1+4+\beta+g_2+g_3}{(\alpha^2-\alpha-1)y_1+2+g_2+g_3}.
\end{align*}
Implying inequalities is pretty much the same as before.\ccc
To make the first fraction minimum, 
$x_1$ should be minimum since it appears only in the numerator and $x_2$ and $y_2$ maximum since they appear in both. 
And, again, since
$\frac{(\alpha^2-\alpha)}{\alpha^2-\alpha-1}\leq2$, 
the fraction is minimum for maximum $y_1$.  
Thus the CR is at least
\begin{align*}
C_B= \frac{(\alpha^2-\alpha)g_1+4+\beta+g_2+g_3}{(\alpha^2-\alpha-1)g_1+2+g_2+g_3}.
\end{align*}

For $\beta=0.4745$ and $\alpha=2.1068$, 
our numerical calculation shows that $C_A=2.107447$ and $C_B=2.106989$, 
and the theorem is proved.\ccc
}

\end{document}